\documentclass[sigconf]{aamas}

\usepackage{balance} 
\usepackage{amsmath}
\usepackage{mathtools}
\usepackage{amsthm}
\usepackage{subcaption}
\usepackage{dsfont}
\usepackage[capitalize,noabbrev]{cleveref}
\usepackage{thmtools}
\usepackage{apxproof}
\usepackage{algorithm}
\usepackage{algorithmic}
\usepackage{multirow}
\usepackage{pifont}
\theoremstyle{plain}
\newtheorem{theorem}{Theorem}[section]

\newtheorem{lemma}[theorem]{Lemma}

\theoremstyle{definition}
\newtheorem{definition}[theorem]{Definition}
\newtheorem{assumption}[theorem]{Assumption}
\theoremstyle{remark}

\newtheorem{example}[theorem]{Example}
\newenvironment{smalleralign}[1][\small]
 {\par\nopagebreak\leavevmode\vspace*{-\baselineskip}%
  \skip0=\abovedisplayskip
  #1%
  \def\maketag@@@##1{\hbox{\m@th\normalfont\normalsize##1}}%
  \abovedisplayskip=\skip0
  \align}
 {\endalign\ignorespacesafterend}

\setcopyright{ifaamas}
\acmConference[AAMAS '25]{}{}
{}{}
\copyrightyear{}
\acmYear{}
\acmDOI{}
\acmPrice{}
\acmISBN{}

\title{Mean Field Correlated Imitation Learning}

\author{Zhiyu Zhao}
\affiliation{
  \institution{Institute of Automation, CAS\\School of Artificial Intelligence, UCAS}
  \city{Beijing}
  \country{China}}
\email{zhaozhiyu22@ia.ac.cn}

\author{Qirui Mi}
\affiliation{
  \institution{Institute of Automation, CAS\\School of Artificial Intelligence, UCAS}
  \city{Beijing}
  \country{China}}
\email{miqirui2021@ia.ac.cn}

\author{Ning Yang}
\affiliation{
  \institution{Institute of Automation, CAS\\School of Artificial Intelligence, UCAS}
  \city{Beijing}
  \country{China}}
\email{ning.yang@ia.ac.cn}

\author{Xue Yan}
\affiliation{
  \institution{Institute of Automation, CAS\\School of Artificial Intelligence, UCAS}
  \city{Beijing}
  \country{China}}
\email{yanxue2021@ia.ac.cn}

\author{Haifeng Zhang}
\affiliation{
  \institution{Institute of Automation, CAS\\School of Artificial Intelligence, UCAS}
  \city{Beijing}
  \country{China}}
\email{haifeng.zhang@ia.ac.cn}

\author{Jun Wang}
\affiliation{
  \institution{University College London}
  \city{London}
  \country{United Kingdom}}
\email{jun.wang@cs.ucl.ac.uk}

\author{Yaodong Yang}
\affiliation{
  \institution{Peking University}
  \city{Beijing}
  \country{China}}
\email{yaodong.yang@pku.edu.cn}

\begin{abstract}
We investigate multi-agent Imitation Learning (IL) within the framework of Mean Field Games (MFGs). Existing MFG IL algorithms assume that demonstrations are sampled from Mean Field Nash Equilibria (MFNE), which limits their adaptability to real-world scenarios. For instance, in traffic networks, public routing recommendations introduce correlated signals, a complexity that MFNE and other existing correlated equilibrium concepts do not capture. To bridge this gap, we propose the Adaptive Mean Field Correlated Equilibrium (AMFCE), a novel equilibrium concept that incorporates correlated signals. We establish the existence of AMFCE under mild conditions and demonstrate that MFNE is a subclass of AMFCE. Building on AMFCE, we introduce Mean Field Correlated Imitation Learning (MFCIL), a more comprehensive MFG IL framework, and provide theoretical guarantees on the quality of the recovered policy. Experimental results, including a real-world traffic flow prediction problem and simulations using the TaxAI environment, demonstrate the superiority of MFCIL over state-of-the-art IL baselines, showcasing its potential to better understand large population behavior under correlated signals.
\end{abstract}


\newcommand{\BibTeX}{\rm B\kern-.05em{\sc i\kern-.025em b}\kern-.08em\TeX}

\begin{document}


\pagestyle{fancy}
\fancyhead{}


\maketitle 


\section{Introduction}
Imitation Learning (IL) is a powerful framework to imitate expert policies from demonstrations \cite{DBLP:journals/csur/HusseinGEJ17}. 
However, in scenarios involving a large population of agents, existing IL algorithms face limitations due to the exponential increase in interactions and dimensionality, limiting their applicability in real-world situations including traffic management \cite{bazzan2009opportunities}, ad auctions \cite{guo2019learning} and economic activities among human \cite{DBLP:conf/atal/MiXS0ZW24}. 
Mean field theory offers a practical alternative to offer an analytically feasible and practically efficient approach for analyzing multi-agent games in systems with homogeneous agents \cite{guo2019learning, DBLP:conf/icml/YangLLZZW18}. 
In Mean Field Game (MFG) settings, the states of the entire population can be effectively summarized into an empirical state distribution due to homogeneity, reducing the problem to a game between a representative agent and an empirical distribution.

The current literature on MFG IL assumes that expert demonstrations are sampled from the classical Mean Field Nash Equilibrium (MFNE) \cite{DBLP:conf/iclr/YangYTXZ18, DBLP:conf/atal/ChenZLH22}. However, this framework lacks the generality needed to accommodate various real-world situations where external correlated signals influence the behavior of the entire population. For instance, this occurs when the decisions of all drivers in a traffic network are influenced by public routing recommendations dependent on the weather. 

Therefore, a more general equilibrium concept is needed before advancing in MFG IL. Inspired by the concept of Correlated Equilibrium (CE) for stateless game \cite{aumann1974subjectivity}, there are recent developments on Mean Field Correlated Equilibrium (MFCE) with state dynamics \cite{campi2022correlated, DBLP:journals/corr/abs-2208-10138}. 
However, existing MFCE concepts \cite{campi2022correlated, DBLP:conf/atal/MullerREPPLMPT22} assume that the fixed correlated signal is sampled from a distribution at the start of the game, allowing agents to observe future information and make decisions accordingly throughout the game. This implies that agents have complete knowledge of future recommendations or signals from the outset. This assumption is impractical in real-world situations where agents must base their decisions on currently available information, without access to future signals. For example, in tax policy implementation, the government (acting as a mediator) adjusts the signal based on evolving economic conditions rather than pre-determined information. In such scenarios, decision-makers (e.g., citizens or businesses) must adaptively respond to current and past signals without knowledge of future information.

In summary, the lack of a general MFG equilibrium concept that accommodates realistic informational constraints, where future correlated signals are unobservable, has limited the applicability of existing models to dynamic, real-world scenarios, thereby impeding the practical application of MFG IL methods. To address the limitations of existing MFCE concepts and MFG IL methods, we introduce a novel equilibrium concept called the ``Adaptive Mean Field Correlated Equilibrium (AMFCE)''. This concept accounts for the inherent uncertainty in real-world environments by assuming that agents can only base their decisions on past and present information, without access to future signals.
Building upon the AMFCE concept, we introduce a new IL framework, namely the ``Mean Field Correlated Imitation Learning (MFCIL)''. This introduction is accompanied by a theoretical guarantee of the quality of the policy recovered by this framework.
The generality and flexibility of AMFCE allow MFCIL framework to predict and explain more real-world scenarios.

Our contributions are summarized as follows:

\begin{itemize}
    \item We propose the concept of AMFCE and establish its existence under mild conditions. Compared with previous MFCE concepts, AMFCE allows agents to operate under more realistic informational constraints, where future signals are unobservable, thereby providing a more practical framework for modeling real-world scenarios.
    We prove that MFNE is a subclass of AMFCE, implying the broader applicability of MFCIL than the existing MFG IL frameworks.
    We provide an example in \cref{subsec:compare} to demonstrate the generality and flexibility of AMFCE over other MFCE concepts.
    Furthermore, we prove that the AMFCE is the limit of CE in the $N-$ player game when the population size approaches infinity. This highlights the broader applicability of AMFCE and the corresponding MFCIL framework over existing MFCE and MFNE concepts, providing a more practical and robust foundation for modeling real-world scenarios with unobservable future signals.
    \item Based on the general AMFCE concept, we propose MFCIL, the first IL framework capable of recovering CE policy in MFG. The inclusion of AMFCE enhances the capabilities of MFCIL, enabling it to surpass MFG IL algorithms based on MFNE, since it can imitate expert policies in a boarder range of scenarios. Moreover, our framework is also suitable for recovering MFNE policy as it is a subclass of AMFCE.
    \item We demonstrate the effectiveness of our proposed framework both theoretically and empirically. 
    Theoretical analysis guarantees the quality of the recovered policy and extends limited existing theoretical results on MFNE to a more general MFG equilibrium. 
    Our framework is {\it the first practical MFG IL framework with a polynomial dependency on the horizon for the performance difference}, surpassing existing practical MFG IL algorithms. 
    Empirical evidence highlights the superiority of our framework over state-of-the-art MFG IL baselines across various tasks, including a real-world traffic flow prediction problem and simulations conducted in the TaxAI environment \cite{DBLP:conf/atal/MiXS0ZW24}.
\end{itemize}

\section{Related work}

\paragraph{Multi-agent Imitation Learning}
Previous research in Multi-agent Imitation Learning (MAIL) has extended single-agent IL algorithms to Markov games \cite{DBLP:conf/nips/SongRSE18, DBLP:conf/icml/YuSE19, jeon2020scalable}. 
However, these algorithms encounter scalability challenges due to the curse of dimensionality. To address the scalability challenge, Yang et al. proposed a multi-type mean field approximation that approximates Nash equilibrium in Markov games \cite{DBLP:conf/nips/YangVC020}. 
Nevertheless, this approach does not consider the MFG and MFNE, thus failing to account for the interdependence between mean field flow and policy.
Yang et al. introduced a method for inferring the MFG model through Inverse Reinforcement Learning (IRL), under the assumption that the equilibrium underlying the demonstrations is the Mean Field Social Optimum (MFSO). This condition is applicable solely to fully cooperative settings \cite{DBLP:conf/iclr/YangYTXZ18}.
Chen et al. extended this method to mixed cooperative-competitive settings by assuming that the demonstrations are sampled from MFNE and its variant \cite{DBLP:conf/atal/ChenZLH22, chen2022agent}. 
Ramponi et al. proposed the solution concept named Nash Imitation Gap (NIG) and provided upper bounds of NIG for several different settings \cite{ramponi2023on}, but they focused on experts achieving a Nash equilibrium. 
\paragraph{Mean Field Equilibria Concepts}
While existing MFG IL algorithms have not incorporated CE, there have been a few, albeit limited, works that introduce CE into the MFG.
Campi and Fischer assume that a mediator recommends the same stochastic policy to the entire population, resulting in a limited equilibrium set identical to the classic MFNE \cite{campi2022correlated}. 
Additionally, it is often more practical for the mediator to recommend actions rather than stochastic policies to individuals.
Muller et al. assume that the mediator recommends a deterministic policy (sampled from a distribution named `population recommendation' over the deterministic policy space) to each individual \cite{DBLP:journals/corr/abs-2208-10138}.
Both MFCE concepts assume that a fixed correlated signal (recommended policy in Campi and Fischer, and population recommendation in Muller et al.) is realized at the beginning of the game, allowing agents to observe future signals or recommendations.
However, this assumption is impractical in real-world scenarios where decisions, such as economic behavior, depend on real-time conditions, with future information remaining inaccessible.
To address these limitations, we propose the AMFCE concept, which extends the existing MFCE framework by allowing agents to operate without access to future signals, making it more applicable to dynamic, real-world environments.

This enhanced flexibility caters to real-world scenarios where varying correlated signals are introduced by the mediator. We also provide a concrete example demonstrating the greater generality of our equilibrium concept over that proposed by Muller et al. \cite{DBLP:journals/corr/abs-2208-10138} in \cref{sec:comparison to muller}. We also discuss the difference between AMFCE and MFNE with common noise \cite{perrin2020fictitious, carmona2016mean} in \cref{sec:common_noise}.

\section{Preliminaries}

\subsection{Classic mean field Nash equilibrium}
The classic MFG models a game between a representative agent and the state distribution of all the other agents.
Denote $\mathcal{P}(\mathcal{X})$ as the set of probability distributions over the set $\mathcal{X}$ and denote $\mathcal{T} = \{0,1,\cdots,T\}$ as a set of time indexes. $T$ is the time horizon. The state space and the action space are denoted as $\mathcal{S}$ and $\mathcal{A}$, respectively. 
The population state distribution of a homogeneous $N$-agent game at time $t$ is $\mu_t(s) \triangleq \lim_{N\rightarrow \infty} \frac{1}{N}\sum_{i=1}^N \mathds{1}{\{s_t^i=s\}}$, where $s_t^i$ is the state of agent $i$ at time $t$, and $\mathds{1}_{\{e\}}$ is an indicator function (with value $1$ if expression $e$ holds and $0$ otherwise). The mean field flow is defined as $\pmb{\mu} = \{\mu_t\}_{t\in\mathcal{T}}$. 
The transition kernel for the state dynamics is denoted as $P:\mathcal{S}\times \mathcal{A}\times \mathcal{P}(\mathcal{S}) \rightarrow \mathcal{P} (\mathcal{S})$. 
At time $t$, after the representative player chooses its action $a_t$ according to policy $\pi_t$, it will receive a deterministic reward $r(s_t,a_t,\mu_t)$, and its state will evolve according to the current state $s_t\in\mathcal{S}$ and transition kernel $P(\cdot|s_t,a_t,\mu_t)$.  

For a fixed mean field flow $\pmb{\mu}$, the objective of the representative agent is to solve the following decision-making problem over all admissible policies $\pmb{\pi}=\{\pi_t\}_{t\in\mathcal{T}}$: 
\begin{align}\label{eq:mfg}
\begin{array}{ll}
\text{maximize}_{\pmb{\pi}} & \left.V_k(s,\pmb{\pi},\pmb{\mu})\triangleq\mathbb{E}\left[\sum\limits_{t=k}^T  \gamma^t r(s_t, a_t, {\color{black}\mu_t})\right|s_k=s\right]\\
\text{subject to} & s_{t+1}\sim P(\cdot|s_t,a_t,{\color{black}\mu_t}),\quad a_t\sim \pi_t(s_t),
\end{array}
\end{align}
where  $\gamma \in (0, 1]$ is the discount factor.
The MFNE \cite{guo2019learning, DBLP:journals/corr/abs-2205-12944} is defined as the following.
\begin{definition}[MFNE]\label{nash2_stat} 
In classic MFG (\cref{eq:mfg}), a policy-population profile ($\pmb{\pi}^\star$, $\pmb{\mu}^\star$) is called an MFNE (under initial state distribution $\mu_0$) if 
\begin{enumerate}
    \item For any policy $\pmb{\pi}$ and any initial state $s\sim \mu_0$, $V_0\left(s,\pmb{\pi}^\star,{\color{black}\pmb{\mu}^\star}\right)\geq V_0\left(s,\pmb{\pi},\pmb{\mu}^\star\right).$

    \item  { (Population side) The mean field flow $\pmb{\mu}^*$ satisfies 
\begin{equation}
    \mu^*_{t}(\cdot) = \sum_{s\in\mathcal{S},a\in \mathcal{A}} P(\cdot|s,a,\mu_{t-1}^*)\pi^*_{t-1}(a|s)\mu^*_{t-1}(s),
\end{equation} 
with initial condition $\mu^*_0=\mu_0$.}
\end{enumerate}
\end{definition}

The single player side condition captures the optimality of $\pmb{\pi}^\star$ when the mean field flow $\pmb{\mu}$ is fixed. The population side condition ensures the ``consistency" of the solution by guaranteeing that the state distribution flow of the single player matches the mean field flow $\pmb{\mu}^{\star}$.  

\subsection{Imitation Learning}
Let $\mathcal{M} = (\mathcal{S}, \mathcal{A}, P, r, \mu_{0}, \gamma, T)$ represent a single-agent Markov decision process (MDP). In this notation, $\mathcal{S}$ and $\mathcal{A}$ denote the state and action spaces, respectively. The transition kernel for the state dynamics is denoted by $P: \mathcal{S} \times \mathcal{A} \rightarrow \mathcal{P}(\mathcal{S})$. The reward function is denoted as $r: \mathcal{S} \times \mathcal{A} \rightarrow \mathbb{R}$. The initial distribution of the initial state $s_0$ is denoted as $\mu_{0}$. The discount factor is represented by $\gamma \in (0, 1]$, and $T$ corresponds to the horizon.
The expected return of a policy $\pi$ is defined as $J(\pi) = \mathbb{E}\left[\sum_{t=0}^{T} \gamma^{t} r(s_{t}, a_{t})\right]$, where the expectation is taken with respect to $s_0 \sim \mu_{0}$, $a_t \sim \pi(\cdot|s_t)$ and $s_{t+1} \sim P(\cdot|s_t, a_t)$.

In the IL setting, a set of expert demonstrations sampled from expert policy $\pi^{E}$ is provided. The goal of IL is to recover the expert policy $\pi^{E}$ using the expert demonstration.

IRL is a subclass of IL and it solves the problem in two steps. It first finds a reward function $\tilde{r}=\max_{r}\big(\min_{\pi}-H(\pi)-J(\pi)\big)+J(\pi^{E})$ that rationalizes the expert policy $\pi^{E}$, where $H(\pi)\triangleq\mathbb{E}_\pi[-\log\pi(a|s)]$ is the causal entropy of the policy $\pi$ \cite{DBLP:conf/cdc/BloemB14}.
Then a recovered policy is learned from the reward function $\tilde{r}$ by a reinforcement learning method.

Generative Adversarial Imitation Learning (GAIL) \cite{DBLP:conf/nips/HoE16} treats IL as a mini-max game and is trained using a Generative Adversarial Network (GAN). GAIL introduces a discriminator $D_{\omega}$ to differentiate state-action pairs from $\pi^{E}$ and other policies. The recovered policy $\pi_{\theta}$, parameterized by $\theta$, plays the role of the generator. 
It aims at generating state-action pairs that are difficult for $D_{\omega}$ to differentiate. The objective function of GAIL is thus defined as
\begin{equation}
    \max_\theta \min _w \mathbb{E}_{(s, a) \sim \pi_\theta}\left[\log \left(D_{\omega}(s, a)\right)\right]+\mathbb{E}_{(s, a) \sim \pi^{E}}\left[\log \left(1-D_{\omega}(s, a)\right)\right],
\end{equation}
where $\mathbb{E}_{(s, a)\sim\pi_{\theta}}$ is expectation taken with respect to $s_{t+1}\sim P(\cdot|s_{t}, a_{t})$, $a_{t}\sim\pi_{\theta}(\cdot|s_{t})$, $s_{0}\sim\mu_{0}$ and $\mathbb{E}_{(s, a)\sim\pi^{E}}$ is expectation taken with respect to $s_{t+1}\sim P(\cdot|s_{t}, a_{t})$, $a_{t}\sim\pi^{E}(\cdot|s_{t})$, $s_{0}\sim\mu_{0}$.

\section{Problem formulation}\label{sec:amfce}
In this section, we introduce the AMFCE and compare AMFCE with existing MFCE concepts.
Then we establish the existence of AMFCE under mild conditions and demonstrate that the solution set of AMFCE is richer than the well-known MFNE. 
We also prove that AMFCE in the mean field game approximates the correlated equilibrium in the finite agent setting in \cref{sec:relation}.
\subsection{Adaptive Mean Field Correlated Equilibrium}
Before the introduction of the AMFCE, we first introduce the concepts of correlation device \cite{DBLP:conf/atal/MullerREPPLMPT22} and behavioral policy.
\begin{definition}[Correlation Device]
    The per-step correlation device $\rho_t\in \mathcal{P}(\mathcal{Z})$ is a distribution over the finite correlated signal space $\mathcal{Z}$, from which the correlated signal $z_t$ is sampled at time $t$. We denote $\pmb{\rho}=\{\rho_t\}_{t=0}^{T}$ as correlation device over the entire horizon.
\end{definition}

\begin{definition}[Behavioral Policy]
    For each time $t$, the per-step behavioral policy $\pi_t: \mathcal{Z} \times \mathcal{S} \rightarrow \mathcal{P}(\mathcal{A})$ maps the state $s$ and correlated signal $z$ to a distribution over the action space $\mathcal{A}$. 
    We denote $\pmb{\pi} = \{\pi_t\}_{t=0}^T$ as the behavioral policy over the entire horizon. The term `policy' may be used to replace `behavioral policy' without confusion.
\end{definition}
 
At each time step $t$, a correlated signal $z_t$ is sampled from the per-step correlation device $\rho_t$. Subsequently, for each agent at state $s_t$, a mediator independently samples an action $a_t$ from the per-step behavioral policy $\pi_t(\cdot|s_t, z_t)$ as the recommended action for the agent. Importantly, this recommended action $a_t$ is {\it private}, accessible only to the respective agent.
Mathematically, denote $\mathcal{I}_t =\{ \rho_t, a_t, \pi_t, s_{t}, \mu_{t}\}$ as the information available to the agent at the beginning of step $t$.\footnote{$\mathcal{I}_t$ serves as a criterion for evaluating whether a policy and correlation device constitute an AMFCE, similar to how the population distribution is used in typical MFNE concepts. The presence of $\mu_t$ does not imply agents have knowledge of the population distribution. Neither the policy $\pi(a|s, z)$ nor the correlation device $\rho(z)$ relies on precise population distribution information.} Note that the agent only observes the functional form of $\pi_t$ but {\it cannot observe} the correlated signal $z_t$ nor the recommended actions for other agents.  
Therefore, the agent has to {\it predict} the correlated signal $z_t$ based on the local information $\mathcal{I}_t$:
\begin{equation}\label{predict_z}
    \rho^{\rm pred}_t(z_t=z|\mathcal{I}_t)=\frac{\rho_t(z)\pi_t(a_t|s_t, z)}{\sum_{z'\in\mathcal{Z}}\rho_t(z')\pi_t(a_t|s_t, z')}.
\end{equation}
The agent can then update the prediction for the population state distribution of the next time step for each possible signal $z$ using the McKean-Vlasov equation:
\begin{equation}
    \begin{aligned}
        \mu^{\rm pred}_{t+1}(\cdot|\mathcal{I}_t, z) &= \sum_{a \in \mathcal{A}}\sum_{s\in \mathcal{S}} \mu_{t}(s)P(\cdot|s,a,\mu_{t})\pi_{t}(a|s,z)\triangleq \Phi(\mu_{t}, \pi_{t},z).
    \end{aligned}
\end{equation}
Given the population state distribution $\mu$, the agent will choose action $a$ to maximize the action value function $Q_{t}^{\pmb{\pi}}(s, a, \mu, z; \pmb{\pi}')$:
\begin{equation}
    Q_{t}^{\pmb{\pi}}(s, a, \mu, z; \pmb{\pi}')
    = r(s, a, \mu) + \gamma\mathbb{E}_{\pmb{\pi}, \pmb{\pi}', \pmb{\rho}}\bigg[\sum_{i=t+1}^{T}\gamma^{i-t-1}r(s_i, a_i, \mu_i)\bigg].
\end{equation}
The action value function is the expected return of an agent when the agent follows policy $\pmb{\pi}$ while the population adheres to policy $\pmb{\pi}'$ under the correlation device $\pmb{\rho}$, conditioned on $(s_t, a_t, \mu_t, z_t)=(s, a, \mu, z)$.
Unless otherwise stated, the expectation $\mathbb{E}_{\pmb{\pi}, \pmb{\pi}', \pmb{\rho}}$ is taken with respect to $z_t\sim\rho_t(\cdot)$, $s_t\sim P(\cdot|s_{t-1}$, $a_{t-1}, \mu_{t-1}$), $a_t\sim\pi_t(\cdot|s_t, z_t)$, $\mu_t=\Phi(\mu_{t-1}, \pi_{t-1}', z_{t-1})$.


To introduce the concept of AMFCE, we define the set of swap function $$\mathcal{U}\triangleq\{u:\mathcal{A}\to\mathcal{A}\},$$ namely $u$ is a function that modifies an action $a$ to an action $u(a)$. 
Let $\Delta_t(s, \mu, u; \pmb{\pi}, \pmb{\rho})=\mathbb{E}\big[Q_t^{\pmb{\pi}}(s, u(a), \mu, z; \pmb{\pi}) - Q_t^{\pmb{\pi}}(s, a, \mu, z; \pmb{\pi})\big]$ denote the difference in the action value function when the agent takes action $u(a)$ in response to a recommendation $a$, where $u\in\mathcal{U}$. The expectation is taken with respect to $z\sim\rho_t(\cdot)$, $a\sim\pi_t(\cdot|s, z)$.
\begin{definition}[AMFCE]\label{defMFCE}
    The profile $(\pmb{\pi}^{\star}, \pmb{\rho}^{\star})$, comprising the behavioral policy $\pmb{\pi}^{\star}=\{\pi_t^{\star}\}_{t=0}^{T}$ and the time-varying correlation device $\pmb{\rho}^{\star}=\{\rho_t^{\star}\}_{t=0}^{T}$, is an AMFCE if
    \begin{enumerate}
        \item (Single agent side) No agent has an incentive to unilaterally deviate from the recommended action after predicting the $z$ by \cref{predict_z}, i.e.
        $
            \Delta_t(s, \mu_t^{\star}, u; \pmb{\pi}^{\star}, \pmb{\rho}^{\star})\le0, \quad 
            \forall u\in \mathcal{U}, \forall s\in\mathcal{S}, \forall t\in\mathcal{T}.
        $
        \item (Population side) 
        The mean field flow $\pmb{\mu}^*$ satisfies $\mu_{t}^{*}(\cdot)=\Phi(\mu_{t-1}^{\star}, \pi_{t-1}^{\star}, z_{t-1})$, given the correlated signals $\{z_{t}\}_{t=0}^{T}$ and initial condition $\mu_{0}^{*}=\mu_{0}$.
    \end{enumerate}
\end{definition}
\subsection{Difference between AMFCE and MFCE}\label{subsec:compare}
In MFCE, the correlated signal is typically assumed to be fully observable by the agents at the start of the game. The policy corresponding to this correlated signal sequence is then recommended to each agent. meaning that agents have access to the entire sequence of signals across the game’s time horizon. This implies that agents can foresee all future signals from the outset. Consequently, the agent can infer or observe the correlated signal $z$ at the start of the game without the need for adaptive updates to its belief. In contrast, AMFCE assumes that agents can only observe past signals, with no access to future signals. This assumption better reflects real-world scenarios where decision-makers must react to unfolding information rather than relying on complete foresight. In the AMFCE framework, correlated signals are realized at each time step. Following the sampling of the correlated signal $z_t$ at time $t$ from the time-varying correlation device $\rho_t$, the action $a_t$ is sampled from the policy $\pi_t(a_t|s_t, z_t)$ for each agent at state $s_t$, serving as a private recommendation. Agents can observe only the recommended action $a_t$ and cannot directly observe the correlated signal $z_t$. Since the correlated signal $z_t$ cannot be realized until time $t$, the agent must \textit{adaptively} update its belief in the correlated signal.
Below, we provide an example to clarify why AMFCE is more practical than existing MFCE concepts.

\begin{table}[t]
    \caption{Comparison of Equilibrium Concepts}
    \label{tab:equilibrium_comparison}
    \centering
    \begin{tabular}{ccc}
      \hline
      \multirow{2}{*}{Equilibrium Concept} & Supports & No Access to \\
      & Correlated Signals & Future Information \\ 
      \hline
      MFNE & \ding{55} & \ding{55} \\ 
      MFCE & \ding{51} & \ding{55} \\
      AMFCE & \ding{51} & \ding{51} \\
      \hline
    \end{tabular}
\end{table}

\begin{example}\label{eg:finite}
    A traffic network comprises three cities. Tourists located in city $C$ are expected to visit city $L$ or $R$ during a two-day vacation period. These tourists rely on an online mapping application that suggests either city $L$ or $R$ based on real-time weather information $z$. This scenario can be modeled as a MFG with a state space $\mathcal{S}=\{C, L, R\}$ and an action space $\mathcal{A}=\{L, R\}$. The initial population state distribution is given by $\mu_0(C)=1$, and the reward function is defined as $r(s, a, \mu)=\mathds{1}_{\{s=L\}}\mu(L)+\mathds{1}_{\{s=R\}}\mu(R)$.
    Due to the possibility of unexpected road closures, the environment transition kernel is non-deterministic.
    The environment transition kernel is shown in the \cref{tab:transition}.
\end{example}
    
\begin{table}
    \caption{The transition probability $P(s_{t+1}|s_t, a_t)$ in the \cref{eg:finite}. $P(s_{t+1}=R|s_t, a)=1-P(s_{t+1}=L|s_t, a)$.}    
    \label{tab:transition}
    \begin{center}
        \begin{tabular}{llllllll}
        \hline
        \multicolumn{2}{l}{$s_t$}                 & $C$ & $C$ & $L$ & $L$ & $R$ & $R$ \\ \hline
        \multicolumn{2}{l}{$a_t$}                   & $L$ & $R$ & $L$ & $R$ & $L$ & $R$ \\ \hline
        \multicolumn{2}{l}{$P(s_{t+1}=L|s_t, a_t)$} & 1   & 0   & 1   & 1/4 & 3/4 & 0   \\ \hline
        \end{tabular}
    \end{center}
\end{table}
\begin{table}
    \caption{The AMFCE policy in the \cref{eg:finite}. $\pi(a=R|s, z)=1-\pi(a=L|s, z)$.}
    \label{tab:policy}
    \begin{center}
        \begin{tabular}{clcccccc}
            \hline
            \multicolumn{2}{c}{$s$}             & $C$   & $C$   & $L$ & $L$   & $R$   & $R$ \\ \hline
            \multicolumn{2}{c}{$t$}             & $0$   & $0$   & $1$ & $1$   & $1$   & $1$ \\ \hline
            \multicolumn{2}{c}{$z_t$}             & $0$   & $1$   & $0$ & $1$   & $0$   & $1$ \\ \hline
            \multicolumn{2}{c}{$\pi_t(a=L|s, z_t)$} & $2/3$ & $1/3$ & $1$ & $1/9$ & $8/9$ & $0$ \\ \hline
        \end{tabular}
    \end{center}
\end{table}
The online mapping application recommends a city for each agent to visit in the following way.
At time $t\in\mathcal{T}=\{0, 1\}$, a correlated signal $z$ is sampled from the correlated signal space $\mathcal{Z}=\{0, 1\}$ with equal probabilities, i.e., $\rho_t(z=0)=\rho_t(z=1)=0.5$. The online mapping application recommends an action for each agent based on the observed value of $z$ and the behavioral policy $\pmb{\pi}$ defined in the \cref{tab:policy}.
It can be verified that tourists have no incentive to deviate from the recommendation, so an AMFCE is achieved. 


Under the MFCE setting, the correlated signal sequence is fully observable. If the sequence is \(\pmb{z} = \{0, 1\}\), the agent can achieve a higher return by selecting $a_0 = L$ and $a_1 = R$, regardless of the actions recommended. Consequently, this example cannot be adequately explained by existing MFCE concepts. Having observed the policy recommended for the population, the agent will take the best response of the recommended policy for achieving a higher return. Therefore, for a given correlated signal sequence, the policy set of MFCE is equivalent to the MFNE policy set, which limits its generality. As shown in \cref{corr:relation}, the MFNE policy set is a subset of the AMFCE policy set, demonstrating that the MFCE framework is more restrictive.

In contrast, AMFCE allows for a broader range of strategies by incorporating adaptive decision-making under the constraint of unobservable future signals. This flexibility enables AMFCE to better model dynamic environments and multi-agent interactions, making it applicable to a wider variety of real-world scenarios.
\subsection{Properties of AMFCE}
This subsection focuses on the properties of AMFCE, including the conditions to guarantee its existence and its relationship to classic MFNE.
To provide the existence of AMFCE solutions, we define the best response operator 
$$
    \mathrm{BR}(\pmb{\pi}; \pmb{\rho}) = \mathop{\arg\max}_{\pmb{\pi}'}\mathbb{E}_{\pmb{\pi}', \pmb{\pi}, \pmb{\rho}}\left[\sum_{t=0}^T\gamma^t r(s_t, a_t, \mu_t)\right].
$$
Then the existence of AMFCE is derived using Kakutani’s fixed point theorem \cite{kakutani1941generalization}  with the operator $\mathrm{BR}$.
We next provide a sufficient condition for the existence of AMFCE. 
\begin{restatable}{theorem}{fixed}\label{thm:fixed}
    If the reward functions $r(s,a,\mu)$ and transition kernel $P(s'|s, a, \mu)$ are bounded and continuous with respect to population state distribution $\mu$, there exists at least one AMFCE solution.
\end{restatable}


AMFCE is a more general equilibrium concept compared to MFNE. \cref{corr:relation} shows that MFNE is a subclass of AMFCE. 
\begin{restatable}{corollary}{relation}\label{corr:relation}
    Every MFNE can be transformed into an AMFCE.
\end{restatable}
The proof is deferred to \cref{prfrelation}. \cref{corr:relation} implies that any IL algorithm designed to recover AMFCE policies can also recover MFNE policies.

\section{Imitation learning for AMFCE}\label{sec:algo}
In this section, we propose a novel IL framework for recovering AMFCE from expert demonstrations.
In the setting of IL, the reward signal is inaccessible.
To construct a suitable reward function rationalizing the expert policy, we define an AMFCE inverse reinforcement learning (AMFCE-IRL) operator to design a reward function from expert demonstrations. 
We denote the AMFCE under the designed reward function $r$ and correlation device $\pmb{\rho}$ as $\mathrm{AMFCE}(r, \pmb{\rho})$. 
The condition of AMFCE, as defined in \cref{defMFCE}, implies that agents cannot improve the policy $\pmb{\pi}$ through 1-step temporal difference learning. We proceed to derive equivalent constraints for multi-step temporal difference learning, outlined in \cref{tstep}.
Utilizing the Lagrangian reformulation of these equivalent multi-step constraints, we propose the IL framework for recovering AMFCE. 
We introduce the concept of the Correlated Imitation Gap (CIP) for deriving the multi-step constraints.
\begin{definition}[CIP]
    For a given action sequence $a_{0:T}$, the policy $\pmb{\pi}$ and correlation device $\pmb{\rho}$, the CIP is defined as 
    $
        \mathcal{R}(a_{0:T}, \pmb{\pi}, \pmb{\rho}) \triangleq \mathbb{E}\Big[\sum_{t=0}^{T} \gamma^t r(s_t, a_t, \mu_{t})\Big|a_{0:T}\Big] - J(\pmb{\pi}, \pmb{\pi}, \pmb{\rho}),    
    $
    where the expectation is taken with respect to $z_t\sim\rho_t(\cdot)$, $s_t\sim P(\cdot|s_{t-1}, a_{t-1}, \mu_{t-1})$.
    Here, $J(\pmb{\pi}, \pmb{\pi}', \pmb{\rho})=\mathbb{E}_{\pmb{\pi}, \pmb{\pi}', \pmb{\rho}}\left[\sum_{t=0}^T\gamma^t r(s_t, a_t, \mu_t)\right]$ represents the expected return of the agent when it follows policy $\pmb{\pi}$ while the population adheres to policy $\pmb{\pi}'$ under the correlation device $\pmb{\rho}$.
\end{definition}
The CIP is defined as the gap of expected return between the agent taking action sequence $a_{0:T}$ and the policy $\pmb{\pi}$.
Then we can get a criterion for AMFCE based on CIP. 
\begin{restatable}{proposition}{tstep}\label{tstep}
    $(\pmb{\pi}, \pmb{\rho})$ is an AMFCE solution if and only if $$\mathcal{R}(a_{0:T}, \pmb{\pi}, \pmb{\rho})\leq 0,$$ $\forall a_{t}\in\mathcal{A}$, $0\le t\le T$. 
\end{restatable}
The proof is deferred to \cref{prftstep}. Intuitively, \cref{tstep} shows the multi-step constraints for AMFCE.
Therefore, the process of finding AMFCE can be defined as an optimization problem with finite constraints measured by the CIP.
We propose a Lagrangian reformulation to find AMFCE.
    $$
    L(\pmb{\pi}, \pmb{\rho}, \lambda, r) \triangleq\sum_{\tau_k \in \mathcal{D}_E}\lambda(\tau_k)\mathcal{R}(a_{0:T}, \pmb{\pi}, \pmb{\rho}),
    $$
where $\mathcal{D}_E$ is a set of action-signal sequences $\tau_k = \{(a_t, z_t)\}_{t=0}^T$.
We show that the Lagrangian form captures the difference of expected returns between two policies by selecting $\lambda$. 
\begin{restatable}{theorem}{dual}\label{thm:dual}
    For policy $\pmb{\pi}^*$ and correlation device $\pmb{\rho}$, let $\lambda_{\pmb{\pi}^{*}}(\tau_k)=\prod_{t=0}^{T}\rho_t(z_t)\pi_t^*(a_t|s_t, z_t)$ be the probability of generating the sequence $\tau_k$ using policy $\pmb{\pi}^*$ and correlation device $\pmb{\rho}$. Then we have 
    $
        L(\pmb{\pi}, \pmb{\rho}, \lambda_{\pmb{\pi}^{*}},r)=J(\pmb{\pi}^*, \pmb{\pi}, \pmb{\rho})-J(\pmb{\pi}, \pmb{\pi}, \pmb{\rho}).
    $
\end{restatable}
The proof of \cref{thm:dual} is deferred to \cref{prfdual}.
Motivated by \cref{thm:dual}, we introduce the AMFCE-IRL operator $\mathrm{AMFCE-IRL}_\psi$ with a reward regularizer $\psi$.
The AMFCE-IRL operator rationalizes the expert policy $\pmb{\pi}^E$ by maximizing the gap in expected return between the expert policy $\pmb{\pi}^E$ and an alternative policy $\pmb{\pi}$.
\begin{equation}\label{mfirl}
    \mathrm{AMFCE-IRL}_\psi(\boldsymbol{\pi}^E, \pmb{\rho}^E) = \arg\max_{r} \Big(-\psi(r) - \max_{\pmb{\pi}} L(\pmb{\pi}^E, \pmb{\rho}^E, \lambda_{\pmb{\pi}^{*}}, r)\Big),
\end{equation}
where $(\pmb{\pi}^E, \pmb{\rho}^E)$ is the AMFCE from which expert demonstrations are sampled. The regularizer for the reward function is chosen as the adversarial reward function regularizer to avoid overfitting \cite{DBLP:conf/nips/HoE16}.

\begin{equation}
    \begin{aligned}
        \psi_{GA}(r) \triangleq\begin{cases}\mathbb{E}_{\pmb{\pi}, \pmb{\pi}^E, \pmb{\rho}^E}[\sum_{t=0}^T\gamma^{t}g(r(s_t, a_t, \mu_t))] & \text { if } r>0 \\ +\infty & \text { otherwise }\end{cases}
    \end{aligned}
\end{equation}
Here, $g(x)= \begin{cases}-x-\log \left(1-e^{x}\right) & \text { if } x<0 \\ +\infty & \text { otherwise }\end{cases}$.

We recover the AMFCE policy $\mathrm{AMFCE}(\tilde{r}, \pmb{\rho}^E)$ by \cref{eq:combine}, where $\tilde{r}=\mathrm{AMFCE-IRL}(\boldsymbol{\pi}^E, \pmb{\rho}^E)$.
\begin{equation}\label{eq:combine}
    \begin{aligned}
        \mathrm{AMFCE}\circ\operatorname{AMFCE-IRL}_\psi(\pmb{\pi}^E, \pmb{\rho}^E)
        =&
        \mathop{\arg\min}_{\pmb{\pi}} 
        \max_{r} J(\pmb{\pi}^E, \pmb{\pi}^E, \pmb{\rho}^E)\\
        &- J(\pmb{\pi}, \pmb{\pi}^E, \pmb{\rho}^E) - \psi_{GA}(r)
    \end{aligned},
\end{equation}

\begin{restatable}{proposition}{GAIL}

    The objective in \cref{eq:combine} can be reformulated as the following practical objective function: 
    \begin{equation}\label{eq:objective}
        \begin{aligned}
            \min_{\pmb{\pi}}&\max_{\omega}\mathbb{E}_{\pmb{\pi}, \pmb{\pi}^E, \pmb{\rho}^E}\bigg[\sum_{t=0}^{T}\gamma^{t}\log D_\omega(s_t, a_t, \mu_t)\bigg]\\
            +&\mathbb{E}_{\pmb{\pi}^E, \pmb{\pi}^E, \pmb{\rho}^E}\bigg[\sum_{t=0}^{T}\gamma^{t}\log\big(1 - D_\omega(s_t, a_t, \mu_t)\big)\bigg],
        \end{aligned}
    \end{equation}
    where $D_\omega$ represents the discriminator network parameterized with $\omega$, taking $(s_t, a_t, \mu_t)$ as input and producing a real number in the range $(0, 1]$ as output.
    \label{GAIL}
\end{restatable}
The proof is deferred to \cref{prfGAIL}. This proposition shows that the AMFCE policy can be recovered by the GAN.
Note that simply using \cref{eq:objective} to solve AMFCE cannot recover $\pmb{\rho}^{E}$, so we derive $\pmb{\rho}$ using a gradient descent method in the \cref{gradofrho} with proof in \cref{prfgradofrho}.
\begin{restatable}{proposition}{gradofrho}\label{gradofrho}
    If the correlation device $\rho_t^{\phi}$ is parameterized with $\phi$, the gradient to optimize $\phi$ given state $s$ is 
    $$
        \mathbb{E}_{z\sim\rho_t^\phi(\cdot)}\bigg[\nabla_\phi \log \rho_t^\phi(z) \mathbb{E}_{a\sim\pi_t(\cdot|s, z)}Q_t^{\pmb{\pi}}(s, a, \mu, z; \pmb{\pi})\bigg].
    $$
\end{restatable}
\begin{table*}[t!]
    \caption{Results for numerical tasks. The performative difference between the recovered policy and the ground truth policy is measured by log loss under different correlated signals. The number in the bracket is the standard deviation over 3 independent runs.}
    \label{tab:new}
    \centering
    \resizebox{\textwidth}{!}{\begin{tabular}{cccccccc}
        \hline
        Task                                                                                     & \begin{tabular}[c]{@{}c@{}}Correlated \\ Signal\end{tabular}               & \begin{tabular}[c]{@{}c@{}}MFCIL \\ (Our Method)\end{tabular} & MFIRL          & MFAIRL         & \begin{tabular}[c]{@{}c@{}}Logistic \\ Regression\end{tabular}    & Multinomial            & MaxEnt ICE    \\ \midrule
        \multirow{4}{*}{\begin{tabular}[c]{@{}c@{}}Squeeze with \\ $T=\{0, 1, 2\}$\end{tabular}} & $z=0$ & \textbf{0.643 (0.000)}                                        & 1.450 (2.857)  & 4.064 (0.879)  & 4.484 (0.054)          & 0.686 (0.002)          & -             \\
                                                                                                 & $z=1$ & 0.647 (0.003)                                                 & 3.245 (1.650)  & 4.144 (0.629)  & \textbf{0.000 (0.000)} & 2.577 (0.149)          & -             \\
                                                                                                 & $z=2$ & \textbf{0.020 (0.001)}                                        & 1.072 (2.229)  & 6.934 (4.447)  & 7.091 (0.107)          & 0.282 (0.087)          & -             \\
                                                                                                 & $z=3$ & 0.045 (0.005)                                                 & 7.871 (4.368)  & 1.027 (1.279)  & 10.638 (0.163)         & \textbf{0.001 (0.001)} & -             \\ \hline
        \multirow{2}{*}{\begin{tabular}[c]{@{}c@{}}Squeeze with \\ $T=\{0, 1\}$\end{tabular}}    & $z=0$   & \textbf{0.648 (0.002)}                                        & 3.828 (1.582)  & 4.067 (0.088)  & 1.985 (0.165)          & 0.991 (0.102)          & 0.946 (0.073) \\
                                                                                                 & $z=1$   & \textbf{0.638 (0.001)}                                        & 2.009 (1.191)  & 10.074 (0.174) & 2.139 (0.169)          & 2.947 (0.359)          & 0.648 (0.011) \\ \hline
        RPS                                                                                      & $z=0$   & \textbf{1.083 (0.000)}                                        & 7.127 (0.753)  & 3.221 (1.330)  & 4.805 (0.131)          & 5.850 (0.306)          & 1.537 (0.019) \\ \hline
        \multirow{4}{*}{Flock}                                                                   & $z=0$   & 0.002 (0.000)                                                 & 5.591 (0.869)  & 12.430 (2.759) & \textbf{0.000 (0.000)} & 1.383 (0.004)          & -             \\
                                                                                                 & $z=1$   & \textbf{0.016 (0.003)}                                        & 11.687 (1.158) & 13.042 (1.533) & 7.887 (0.031)          & 1.127 (0.007)          & -             \\
                                                                                                 & $z=2$   & \textbf{0.045 (0.009)}                                        & 7.500 (3.955)  & 10.065 (5.074) & 18.339 (0.010)         & 0.951 (0.009)          & -             \\
                                                                                                 & $z=3$   & \textbf{0.026 (0.003)}                                        & 3.847 (3.967)  & 9.312 (4.711)  & 35.253 (0.037)         & 1.264 (0.011)          & -             \\ \hline 
        \end{tabular}}
    \end{table*}

The population state distribution $\mu_t$ influences both the input of $D_\omega$ and transition kernel in \cref{eq:objective}. However, the population state distribution $\mu_t$ in expert demonstrations is often inaccessible. 
We characterize $\mu_t$ using the signature of $\pmb{z}_{0:t}$ from rough path theory \cite{DBLP:conf/iclr/KidgerL21}, denoted as $\hat{\mu}_t=\mathrm{Sig}(\pmb{z}_{0:t})$, bypassing the circular reasoning problem \cite{ramponi2023on}. Please refer to \cref{sec: def-sig} for details.
\begin{algorithm}[htpb]
    \begin{algorithmic}
        \REQUIRE Expert demonstration set sampled from $(\pmb{\pi}, \pmb{\rho})$: $\mathcal{D}_E=\{s_0, z_0, a_0, s_1, z_1, a_1, \dots s_{T}, z_{T}, a_{T}\}$, initial population state distribution $\mu_0$.
        \FOR{each iteration}
        \STATE Obtain trajectories from $(\pmb{\pi}, \pmb{\rho})$ by the process: $s_0 \sim \mu_0$,  $a_t \sim \pi^{\theta}(\cdot \vert s_t, z_t)$, $s_{t+1} \sim P(\cdot\mid s_t, \mu_t)$, $z_t\sim\rho_t^\phi(\cdot)$;
        \FOR{$i$ in $\{0, 1, 2, \dots\}$}
        \STATE Update $\omega$ based on the surrogate objective function \cref{eq:surrogate}.
            \ENDFOR
            \FOR{$t$ in $\{0, 1, 2, \dots\}$}
            \STATE Update $\theta$ by Actor-Critic algorithm with small step size based on the surrogate objective function \cref{eq:surrogate}.
            \STATE Update $\phi$ according to \cref{gradofrho};
            \ENDFOR
            \ENDFOR
            
\STATE \textbf{Return} Policy $\pmb{\pi}^{\theta}$, correlation device $\pmb{\rho}^{\phi}$.
    \end{algorithmic}
    \caption{Mean field correlated imitation learning (MFCIL)}
    \label{algo}
\end{algorithm}
We approximately optimize the following surrogate objective function of \cref{eq:objective}.
\begin{equation}\label{eq:surrogate}        
    \begin{aligned}    
        \min_{\pmb{\pi}}&\max_{\omega}\mathbb{E}_{\pmb{\pi}, \pmb{\pi}, \pmb{\rho}^E}\bigg[\sum_{t=0}^{T}\gamma^{t}\log D_\omega(s_t, a_t, \hat{\mu}_t)\bigg]\\
        &+\mathbb{E}_{\pmb{\pi}^E, \pmb{\pi}^E, \pmb{\rho}^E}\bigg[\sum_{t=0}^{T}\gamma^{t}\log\big(1 - D_\omega(s_t, a_t, \hat{\mu}_t)\big)\bigg],
    \end{aligned}
\end{equation}
Combine the above analysis, we propose a new framework, MFCIL, to recover the AMFCE policy and the correlation device from expert demonstrations. The algorithm is shown in \cref{algo}.
Although this framework is designed for recovering AMFCE, it can also be applied to recover MFNE by setting the correlation device as Dirac distribution.
In the \cref{thm: bound}, we provide a theoretical guarantee for the quality of the policy recovered by MFCIL.
\begin{assumption}\label{assp:bound}
    The transition kernel $P(\cdot|s, a, \mu)$ and the reward function $r(s, a, \mu)$ are Lipschitz continuous with respect to population state distribution $\mu$ and have corresponding Lipschitz constants $L_P$ and $L_R$, respectively.
    The reward function is bounded by $r_{\max}$.
    The expert policy $\pmb{\pi}^E$ and recovered policy $\pmb{\pi}$ satisfy
    \begin{equation}
        \begin{aligned}
            \max_{\omega}\mathbb{E}_{\pmb{\pi}, \pmb{\pi}, \pmb{\rho}^E}&\bigg[\sum_{t=0}^{T}\gamma^{t}\log D_\omega(s_t, a_t, \hat{\mu}_t)\bigg]\\
            &+\mathbb{E}_{\pmb{\pi}^E, \pmb{\pi}^E, \pmb{\rho}^E}\bigg[\sum_{t=0}^{T}\gamma^{t}\log\big(1 - D_\omega(s_t, a_t, \hat{\mu}_t)\big)\bigg]\le \epsilon,
        \end{aligned}
    \end{equation}
    which can be achieved by MFCIL.
\end{assumption}
\begin{restatable}{theorem}{bound}
    Under \cref{assp:bound}, for a given action sequence $a_{0:T}$, the CIP of recovered policy $\pmb{\pi}$ is bounded by $
        \mathcal{R}(a_{0:T}, \pmb{\pi}, \pmb{\rho}^E)\le 2\left(2L_R+r_{\max}+\gamma TL_Pr_{\max}\right)\sqrt{2\epsilon T}.
    $.
    \label{thm: bound}
\end{restatable}
The proof is deferred to \cref{prf: bound}.
As the value of $\epsilon$ decreases, the policy $\pmb{\pi}$ recovered by MFCIL approaches the AMFCE policy more closely.
If $\epsilon=0$, the recovered policy $\pmb{\pi}$ is an exact AMFCE policy.
We also provide the imitation gap between the recovered policy in \cref{corr: imitation-gap}.
\begin{restatable}{corollary}{imitationgap}
    The imitation gap \cite{ramponi2023on} between the recovered policy $\pmb{\pi}$ is bounded by
    $
        \max_{\hat{\pmb{\pi}}}J(\hat{\pmb{\pi}}, \pmb{\pi}, \pmb{\rho}^E)-J(\pmb{\pi}, \pmb{\pi}, \pmb{\rho}^E)\le2(3L_R+\gamma TL_Pr_{\max}+r_{\max})\sqrt{2\epsilon T}
    $
    \label{corr: imitation-gap}
\end{restatable}
The proof is deferred to \cref{prf: imitation-gap}.
The imitation gap in \cref{corr: imitation-gap} exhibits a polynomial dependency on the horizon.
The analysis of Ramponi et al. \cite{ramponi2023on} implies that the imitation gap between the recovered policy and the AMFCE policy has an exponential dependency on the horizon for existing practical MFG IL methods.
Therefore, our proposed MFCIL framework represents an improvement over existing practical MFG IL methods.
\section{Experiments}\label{sec:experiment}
\begin{table}[t]
    \caption{Results of predicted traffic flow for Traffic Network. The metric is log loss. The number in the bracket is the standard deviation over 3 independent runs.}
    \label{tab:traffic_network_prediction_table}
    \center
    \resizebox{0.45\textwidth}{!}{\begin{tabular}{cccc}
        \hline
                              & Lewisham               & Hammersmith            & Ealing                 \\ \hline
        MFCIL  (Our Method)   & \textbf{0.742 (0.011)} & \textbf{0.897 (0.002)} & \textbf{1.091 (0.001)} \\ \hline
        MFIRL                 & 12.346 (0.294)         & 9.853 (2.892)          & 11.625 (0.435)         \\ \hline
        MFAIRL                & 8.893 (2.302)          & 6.485 (1.940)          & 11.609 (1.202)         \\ \hline
        \multicolumn{1}{l}{} & Redbridge              & Enfield                & Big Ben                \\ \hline
        MFCIL  (Our Method)   & \textbf{0.052 (0.011)} & \textbf{0.394 (0.003)} & \textbf{1.599 (0.000)} \\ \hline
        MFIRL                 & 11.720 (0.633)         & 11.750 (0.603)         & 7.482 (1.539)          \\ \hline
        MFAIRL                & 4.537 (4.544)          & 9.871 (4.052)          & 12.477 (1.005)         \\ \hline
    \end{tabular}}
\end{table}
\begin{table}[t]
\caption{Results of TaxAI Environment.}
\label{tab:taxai}
\resizebox{0.45\textwidth}{!}{\begin{tabular}{cccc}
\hline
                                                                                       & MFCIL                   & MFIRL          & MFAIRL         \\ \hline
\begin{tabular}[c]{@{}c@{}}Wasserstein Distance \\ with the Expert Policy\end{tabular} & \textbf{26.681} & 33.076 & 49.532 \\ \hline
Household Reward                                                                       & \textbf{29243.837}      & 75.776         & 116.803        \\ \hline
\multicolumn{1}{l}{Government Reward}                                                  & \textbf{3355.917}       & -678.563       & -417.833 \\ \hline     
\end{tabular}}
\end{table}
\begin{figure}[htbp!]
    \centering
    \begin{subfigure}{0.45\linewidth}
        \centering
        \includegraphics[width=\linewidth]{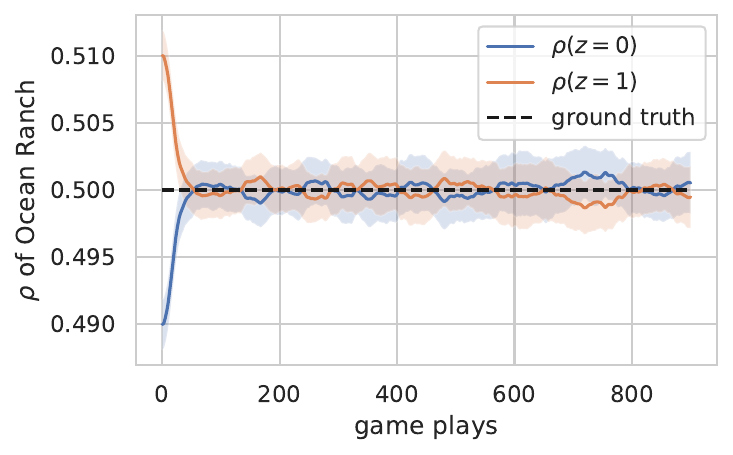}
        \caption{Recovered $\rho$ for Squeeze}
    \end{subfigure}
    \begin{subfigure}{0.45\linewidth}
        \centering
        \includegraphics[width=\linewidth]{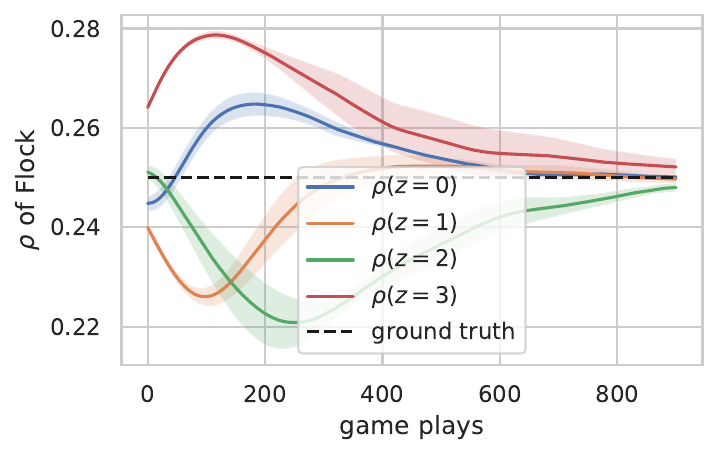}
        \caption{Recovered $\rho$ for Flock}
    \end{subfigure}
    \caption{The distribution of correlation device $\rho$ recovered by MFCIL. The solid line shows the mean and the shaded area represents the standard deviation over 3 independent runs. The dash line shows the ground truth of $\rho$.}
    \Description{The distribution of correlation device $\rho$ recovered by MFCIL. The solid line shows the mean and the shaded area represents the standard deviation over 3 independent runs. The dash line shows the ground truth of $\rho$.}\label{fig:rho_fig}
\end{figure}
\subsection{Tasks}
We evaluate MFCIL on widely used MFG tasks: Sequential Squeeze (Squeeze for short), Rock-Paper-Scissors (RPS), Flock, and a real-world traffic flow prediction task.
The first three experiments are numerical experiments. For numerical experiments, the expert policies are solved analytically.
The traffic flow prediction task is to predict the traffic flow in a complex traffic network based on the real-world data.
Given the large-scale and high-complexity nature of this task, we compare the scalability of MFCIL against MFIRL and MFAIRL in this experiment.
More details about the tasks are deferred to \cref{Detail}. 

\paragraph{Squeeze}: Sequential Squeeze is a game with multi-steps. The purpose of implementing this game is to verify the ability to recover expert policy through demonstrations sampled from a multi-step game. 

\paragraph{RPS}:
This task RPS is a traditional MFG task \cite{chen2021agent, cui2021approximately, chen2022agent}. 
The demonstrations are sampled from MFNE.
We use RPS to verify that the algorithm proposed can recover MFNE, which also supports the result in \cref{corr:relation}. 

\paragraph{Flock}:
The task Flock is based on the movement of fish. 
This task aims to evaluate the performance of algorithms in a MFG that does not satisfy the monotonicity condition \cite{DBLP:conf/ijcai/PerrinLPGEP21}.

\paragraph{Traffic Flow Prediction}:
In the Traffic Flow Prediction task, we use the traffic data of London from Uber Movement. The environment dynamic is deterministic. Our goal is to predict traffic flow in a real-world traffic network consisting of six locations: Lewisham, Hammersmith, Ealing, Redbridge, Enfield, and Big Ben. We collected the individual traveling data among these six locations from Uber Movement as expert demonstrations. The traveling data includes origin, destination, and date. The data has already been anonymized. The results are shown in \cref{tab:traffic_network_prediction_table}.

\paragraph{TaxAI}
We utilize the TaxAI environment, a comprehensive multi-agent reinforcement learning simulator that encompasses a wide range of economic activities, including those of governments, households, technological sectors, and financial intermediaries. We leveraged expert demonstrations derived from the statistical data of the 2022 Survey of Consumer Finances\footnote{https://www.federalreserve.gov/econres/scfindex.htm}.
\subsection{Baselines}
We compare our proposed MFCIL framework with state-of-the-art MFG IL algorithms, MFIRL \cite{DBLP:conf/atal/ChenZLH22}, and MFAIRL \cite{chen2022agent}. 
Since MFIRL and MFAIRL do not take the correlated signal into consideration, we treat the signature of correlated signals as an extension of the state for their algorithms, enabling a fair comparison among all methods. 
It is essential to note that our proposed method is the first IL framework to recover both the policy and the correlation device from data. 
However, as MFIRL and MFAIRL can only recover the policy, we assess the quality of the learned policies for all methods. 
Our focus lies in the difference between the recovered policy and the expert policy, as shown in \cref{tab:new} and \cref{tab:traffic_network_prediction_table}, to evaluate the quality of the policy learned by each method.
We also compare MFCIL with MaxEnt ICE, smoothed multinomial distribution over the joint actions, and logistic regression \cite{DBLP:journals/corr/WaughZB13}.
As MaxEnt ICE is designed to recover correlated equilibrium in the matrix game, we only compare MFCIL with MaxEnt ICE on tasks RPS and Sequential Squeeze with $\mathcal{T}=\{0, 1\}$.
We use the log loss, $\mathbb{E}_{a\sim\pi^E(\cdot|s, z)}[-\log(\pi(a|s, z))]$, to measure the difference between the recovered policy $\pi$ and the expert policy $\pi^E$ in all numerical tasks and the Traffic Flow Prediction Task.
The action space in TaxAI simulator is continuous. We employ the Wasserstein distance to measure the discrepancy between the recovered policy $\pi$ and the expert policy $\pi^E$. Additionally, we present both the government reward and household outcomes. The results are displayed in Table \ref{tab:taxai}. 
\subsection{Results and Analysis}
The results for numerical tasks are presented in \cref{tab:new}. Overall, MFCIL consistently outperforms other methods. While supervised learning methods, such as logistic regression and smoothed multinomial distribution, may occasionally surpass MFCIL in certain metrics, they generally suffer from higher log loss compared to MFCIL. MFIRL and MFAIRL exhibit larger deviations and higher log loss than MFCIL in both \cref{tab:new}, \cref{tab:traffic_network_prediction_table} and \cref{tab:taxai}. These results underscore the inability of MFIRL and MFAIRL to recover AMFCE and handle time-varying correlated signals effectively.
MFCIL consistently outperforms MFIRL and MFAIRL in the Squeeze and Flock tasks because it is the first IL framework capable of recovering AMFCE. This capability allows MFCIL to handle more general scenarios where MFNE-based frameworks may struggle. In the RPS task, MFCIL surpasses other algorithms for two key reasons. Theoretically, MFCIL achieves a significantly lower bound on the error between the occupancy measure of the recovered policy and the expert policy compared to traditional MFIL methods. Since MFNE is a subclass of AMFCE, MFCIL naturally outperforms others in this task. Technically, MFCIL leverages the correlated signal sequence to directly characterize the population distribution, bypassing the variance introduced by estimating the population distribution from samples, resulting in lower deviation.
MaxEnt ICE performs poorly due to its limited reward function class, assuming a linear reward structure. \cref{fig:rho_fig} illustrates that MFCIL can recover the correlation device with rapid convergence speed. 
\section{Conclusion}
In this paper, we investigated the problem of IL for MFGs with time-varying correlated signals. We proposed a novel equilibrium concept, AMFCE, which is better suited for real-world scenarios where the behavior of the entire population is influenced by time-varying correlated signals. Based on this equilibrium concept, we introduced a novel IL framework, MFCIL, to recover the AMFCE policy and correlation device from demonstrations.
Theoretically, we proved that the performance difference and imitation gap between the recovered policy and the expert policy are bounded by a polynomial function of the horizon, which is a significant improvement over existing practical MFG IL results. Empirically, we evaluated our method on several tasks, including one from the real world. 
Our experimental results showed that our method outperforms state-of-the-art MFG IL algorithms. These results highlight the potential of our method to predict and explain large population behavior under correlated signals.



\bibliographystyle{ACM-Reference-Format} 
\bibliography{sample}
\appendix

\section{Proof}

\subsection{The relationship between AMFCE and CE}\label{sec:relation}
In this subsection, we prove the relationship between AMFCE and CE.
AMFCE in the mean field game approximates the correlated equilibrium in the finite agent setting.
We first consider the policy swap function $u: \Pi\rightarrow\Pi$, mapping a policy $\pmb{\pi}\in\Pi$ into another policy $u(\pmb{\pi})\in\Pi$.

Beginning with the definition of the AMFCE,
$$\begin{aligned}
\mathbb{E}_{z_T \sim \rho_T, a_T \sim \pi_T(\cdot|s_T, z_T)}\left[r(s_T, a_T', \mu_T) - r(s_T, a_T, \mu_T)\right] \leq 0, \quad \forall a_T' \in \mathcal{A}
\end{aligned}$$
We can deduce that:
$$\begin{aligned}
    \mathbb{E}_{z_T\sim{\rho}_T, a_T\sim\pi_T(\cdot|s_T, z_T), a\sim u(\pi)(\cdot|s_T, z_T)}\left[r(s_T, a, \mu_T)-r(s_T, a_T, \mu_T)\right]\leq 0
\end{aligned}$$
Hence, for $t=T$, the inequality $Q_t^{u(\pmb{\pi})}(s_t, a_t, \mu_t, z; \pmb{\pi})\leq Q_t^{\pmb{\pi}}(s_t, a_t, \mu_t, z; \pmb{\pi})$ holds for $a_t\in\mathcal{A}$.

Assuming this inequality holds for $t=k$, we can derive from the Bellman Equation:

$$
\begin{aligned}
Q_{k-1}^{u(\pmb{\pi})}(s, a, \mu, z; \pmb{\pi})\leq r(s, a, \mu)+\gamma \mathbb{E}\left[Q_k^{u(\pmb{\pi})}(s', a', \Phi(\mu, \pi_{k-1}, z); \pmb{\pi})\right]\\\leq r(s, a, \mu)+\gamma \mathbb{E}\left[Q_k^{\pmb{\pi}}(s', a', \Phi(\mu, \pi_{k-1}, z); \pmb{\pi})\right]=Q_{k-1}^{\pmb{\pi}}(s, a, \mu, z; \pmb{\pi})
\end{aligned}
$$
By induction, we establish:

$$
Q_1^{u(\pmb{\pi})}(s, a, \mu, z; \pmb{\pi})\leq Q_1^{\pmb{\pi}}(s, a, \mu, z; \pmb{\pi}), \forall a\in\mathcal{A}
$$
Therefore, 
$$\mathbb{E}_{u(\pmb{\pi}), \pmb{\pi}, \pmb{\rho}}\left[\sum_{t=1}^{T}\gamma^tr(s_t, a_t, \mu_t)\right]-\mathbb{E}_{\pmb{\pi}, \pmb{\pi}, \pmb{\rho}}\left[\sum_{t=1}^{T}\gamma^tr(s_t, a_t, \mu_t)\right]\leq 0$$

For any mean field game $\mathcal{G}$, we can associate a stochastic game $\mathcal{G}^N$ with $N$ exchangeable players. $\mathcal{G}^N$ shares the same state space, action space, and initial state as $\mathcal{G}$. The behavior of one player in $\mathcal{G}^N$ depends solely on the population state distribution $\bar{\mu}$.

$$\bar{\mu}(s)=\frac1N\sum_{n=1}^N\mathbb{I}_{\{s^i=s\}}$$

The reward function $r^N(s, a, \bar{\mu})$ and transition kernel $P^N(\cdot|s, a, \bar{\mu})$ of $\mathcal{G}^N$ are identical to $\mathcal{G}$.

From the Theorem 3.3.2 \cite{tembine2009mean}, we have that the 
\begin{align*}
    \lim_{N\to\infty}\mathbb{E}_{u(\pmb{\pi}^i), \pmb{\pi}^{-i}, \pmb{\rho}}\left[\sum_{t=1}^{T}\gamma^tr^N(s_t^i, a_t^i, \mu_t)\right]-E_{\pmb{\pi}^i, \pmb{\pi}^{-i}, \pmb{\rho}}\left[\sum_{t=1}^{T}\gamma^tr^N(s_t, a_t^i, \mu_t)\right]\\=E_{u(\pmb{\pi}), \pmb{\pi}, \pmb{\rho}}\left[\sum_{t=1}^{T}\gamma^tr(s_t, a_t, \mu_t)\right]-E_{\pmb{\pi}, \pmb{\pi}, \pmb{\rho}}\left[\sum_{t=1}^{T}\gamma^tr(s_t, a_t, \mu_t)\right]\leq 0
\end{align*}

Therefore, AMFCE in the mean field game $\mathcal{G}$ approximates the correlated equilibrium in the finite agent setting.
\subsection{Proof of Bellman equation}\label{sec:prfbellman}
In this subsection, we prove the Bellman equation of $Q^{\pmb{\pi}}$ and the optimal action value function $Q^*$.
\begin{align}
    Q_{t}^{\pmb{\pi}}(s, a, \mu, z; \pmb{\pi}')
    =r(s, a, \mu)+\gamma\mathbb{E}\bigg[Q_{t+1}^{\pmb{\pi}}(s', a', \Phi(\mu, \pi'_t, z), z'; \pmb{\pi}')\bigg],\nonumber
\end{align}
Here, the expectation is taken with respect to $z'\sim\rho_{t+1}(\cdot)$, $s'\sim P(\cdot|s, a, \mu)$, and $a'\sim\pi_{t+1}(\cdot|s, z')$. This is conditioned on $(s_t, a_t, \mu_t, z_t)=(s, a, \mu, z)$.

\begin{proof}
    \begin{align}\label{eq:Qf1}
    &Q_{t}^{\pmb{\pi}}(s, a, \mu, z; \pmb{\pi}')\\=&r(s, a, \mu)\\
    &+\gamma\mathbb{E}_{\pmb{\pi}, \pmb{\pi}', \pmb{\rho}}\left[\sum_{i=t+1}^{T}\gamma^{i-t-1}r(s_i, a_i, \mu_i)\middle|(s_t, a_t,\mu_t, z_t)=(s, a, \mu, z)\right]\nonumber\\
    =&r(s, a, \mu)+\gamma\mathbb{E}_{\pmb{\pi}, \pmb{\pi}', \pmb{\rho}}\Big[r(s_{t+1}, a_{t+1}, \Phi(\mu, \pi'_t, z))\nonumber\\
    &+\gamma\sum_{i=t+2}^{T}\gamma^{i-t-2}r(s_i, a_i, \mu_i)\Big|(s_t, a_t,\mu_t, z_t)=(s, a, \mu, z)\Big],
\end{align}
From the definition of action value function $Q_t^{\pmb{\pi}}$, we have
\begin{align}\label{eq:Qf2}
    &\mathbb{E}_{\pmb{\pi}, \pmb{\pi}', \pmb{\rho}}\Big[r\big(s', a', \Phi(\mu, \pi'_t, z)\big)+\gamma\sum_{i=t+2}^{T}\gamma^{i-t-2}r(s_i, a_i, \mu_i)\Big]\nonumber\\
    =&\mathbb{E}\bigg[r\big(s', a', \Phi(\mu, \pi'_t, z)\big)\nonumber+\gamma\mathbb{E}_{\pmb{\pi}, \pmb{\pi}', \pmb{\rho}}\big[\sum_{i=t+2}^{T}\gamma^{i-t-2}r(s_i, a_i, \mu_i)\big]\bigg]\nonumber\\
    =&\mathbb{E}\bigg[Q_{t+1}^{\pmb{\pi}}\big(s', a', \Phi(\mu, \pi'_t, z), z'; \pmb{\pi}'\big)\bigg],
\end{align}
where the outer expectation is taken with respect to $z'\sim\rho_{t+1}(\cdot), s'\sim P(\cdot|s, a, \mu), a'\sim\pi(\cdot|s, z)$. The outer expectation is the conditional expectation given $(s_t, a_t, \mu_t, z_t)=(s, a, \mu, z)$ and the inner expectation is the conditional expectation given $(s_{t+1}, a_{t+1}, \mu_{t+1}, z_{t+1})=(s', a', \Phi(\mu, \pi'_t, z), z')$. We omit the conditions $(s_t, a_t, \mu_t, z_t)=(s, a, \mu, z)$ and $(s_{t+1}, a_{t+1}, \mu_{t+1}, z_{t+1})=(s', a', \Phi(\mu, \pi'_t, z), z')$ for brevity.
Combine \cref{eq:Qf1} and \cref{eq:Qf2}, we get the Bellman equation.
\begin{align*}
    &Q_{t}^{\pmb{\pi}}(s, a, \mu, z; \pmb{\pi}')=r(s, a, \mu)+\gamma\mathbb{E}\bigg[Q_{t+1}^{\pmb{\pi}}\big(s', a', \Phi(\mu, \pi'_t, z), z'; \pmb{\pi}'\big)\bigg],
\end{align*}
where expectation is taken with respect to $z'\sim\rho_{t+1}(\cdot), s'\sim P(\cdot|s, a, \mu), a'\sim\pi_t(\cdot|s, z)$.
\end{proof}
Similarly, we define the optimal action value function $Q_t^*(s, a, \mu, z; \pmb{\pi}')$ as the action value function associated with the optimal individual policy $\pmb{\pi}^*$ when population adheres to policy $\pmb{\pi}'$. It is easy to show that $Q^*$ satisfies the following Bellman equation:
\begin{align}
    Q_t^*(s, a, \mu, z; \pmb{\pi}')=r(s, a, \mu)+\gamma\max_{a'\in\mathcal{A}}\mathbb{E}\bigg[Q_{t+1}^*(s', a', \Phi(\mu, \pi'_t, z), z'; \pmb{\pi}')\bigg],
\end{align}
where the expectation is taken with respect to $z'\sim\rho_{t+1}(\cdot), s'\sim P(\cdot\mid s, a, \mu_t)$. This is conditioned on $(s_t, a_t, \mu_t, z_t)=(s, a, \mu, z)$.

It is worth noting that if the policy of population $\pmb{\pi'}$ is fixed, $Q_T^*(s, a, \mu, z; \pmb{\pi}')\geq Q_T^{\pmb{\pi}}(s, a, \mu, z; \pmb{\pi}')$ for any $\pmb{\pi}$. Then by induction, it holds that $Q_t^*(s, a, \mu, z; \pmb{\pi}')\geq Q_t^{\pmb{\pi}}(s, a, \mu, z; \pmb{\pi}')$ for all $t\in \mathcal{T}$.

\subsection{Proof of \cref{thm:fixed}}
\begin{lemma}\label{prop}
    Policy $\pmb{\pi}'$ is the best response of $\pmb{\pi}$ given $\pmb{\rho}$ if and only if 
    $\sum_{z\in\mathcal{Z}}\rho_t(z)\pi'_t(a|s, z) > 0$ is a sufficient condition of $a\in\mathop{\arg\max}_{a' \in \mathcal A} \mathbb{E}_{z\sim\rho^{\rm pred}_t(\cdot|\mathcal{I}_t)}Q^*(s, a', \mu, z; \pmb{\pi})$, $\forall t\in\mathcal{T}$. 
\end{lemma}
\begin{proof}\label{prfprop}
    We denote 
    \begin{align*}
        \mathcal{Q}_t^{\pmb{\pi}}(s, a, \mu, \mathcal{I}_t; \pmb{\pi})=\mathbb{E}_{z\sim\rho^{\rm pred}_t(\cdot|\mathcal{I}_t)}Q_t^{\pmb{\pi}}(s, a, \mu, z; \pmb{\pi})
    \end{align*} 
    and $\mathcal{Q}_t^{*}(s, a, \mu, \mathcal{I}_t; \pmb{\pi})=\mathbb{E}_{z\sim\rho^{\rm pred}_t(\cdot|\mathcal{I}_t)}Q_t^{*}(s, a, \mu, z; \pmb{\pi})$.

    If the policy $\pmb{\pi}'\in\mathrm{BR}(\pmb{\pi}; \pmb{\rho})$, representing the best response of policy $\pmb{\pi}$ given correlation device $\pmb{\rho}$, and the condition $$\sum_{z\in\mathcal{Z}}\rho_{t}(z)\pi't(a|s, z) > 0$$ is not sufficient for $a\in\mathop{\arg\max}{a' \in \mathcal A} \mathcal{Q}t^{*}(s, a, \mu, \mathcal{I}t; \pmb{\pi})$, then there exists a time step $t\in\mathcal{T}$ such that $\sum{z\in\mathcal{Z}}\rho{t}(z)\pi't(a|s, z) > 0$, while $a \not\in \mathop{\arg\max}{a' \in \mathcal A} \mathcal{Q}_t^{*}(s, a', \mu, \mathcal{I}_t; \pmb{\pi})$.
    
    If $\pmb{\pi}$ and $\pmb{\rho}$ are fixed, the mean field flow is also fixed. Finding the best response of $\pmb{\pi}$ is equivalent to solving an MDP.
    Then the expected return is $J(\pmb{\pi}', \pmb{\pi}, \pmb{\rho})=\mathbb{E}\left[\mathcal{Q}_0^{\pmb{\pi}'}(s_0, a_0, \mu_0, \mathcal{I}_0; \pmb{\pi})\right]$, where the expectation is taken with respect to  $z\sim\rho_0(\cdot)$, $s_0\sim\mu_0$, $a_0\sim\pi'_0(\cdot|s_0, z_0)$. 
    
    We assume that there exists $\pmb{\pi}^*$ such that $\sum_{z\in\mathcal{Z}}\rho_{t}(z)\pi_t^*(a|s, z) > 0$ is sufficient condition of $a \in \mathop{\arg\max}_{a' \in \mathcal A} \mathcal{Q}_t^{*}(s, a, \mu, \mathcal{I}_t; \pmb{\pi})$. The expected return of $\pmb{\pi}^{*}$ is higher than the expected return of $\pmb{\pi}'$ as suboptimal action is impossible to be sampled in the MDP under the population policy $\pmb{\pi}$, which conflicts with the assumption.
    
    If there exists $\pmb{\pi}'$ such that for all $a \in \mathop{\arg\max}_{a' \in \mathcal A} \mathcal{Q}_t^{*}(s, a, \mu, \mathcal{I}_t; \pmb{\pi})$, we have $\sum_{z\in\mathcal{Z}}\rho_{t}(z)\pi'_t(a|s, z) > 0$ is true. 
    Using the induction, we have $\mathbb{E}\left[\mathcal{Q}_0^{\pmb{\pi}'}(s_0, a_0, \mu_0, \mathcal{I}_0; \pmb{\pi})\right]=\max_{\pmb{\tilde{\pi}}}\mathbb{E}\left[\mathcal{Q}_0^{\pmb{\tilde{\pi}}}(s_0, a_0, \mu_0, \mathcal{I}_0; \pmb{\pi})\right]$, where the first expectation is taken with respect to  $z\sim\rho_0(\cdot)$, $s_0\sim\mu_0$, $a_0\sim\pi'_0(\cdot|s_0, z_0)$ and the second expectation is taken with respect to  $z\sim\rho_0(\cdot)$, $s_0\sim\mu_0$, $a_0\sim\tilde{\pi}_0(\cdot|s_0, z_0)$. So the $\pmb{\pi}'$ is the best response of $\pmb{\pi}$ given correlation device $\pmb{\rho}$.
\end{proof}
\begin{lemma}\label{closed_graph}
    $\operatorname{BR}(\boldsymbol{\pi}; \pmb{\rho})$ has a closed graph.
\end{lemma}
\begin{proof}
    We assume that $\lim_{n\to\infty}\pmb{\pi}_n = \pmb{\pi}$, $\lim_{n\to\infty}\pmb{\pi}'_n = \pmb{\pi}'$,
    $\pmb{\pi}_n\in \operatorname{BR}(\pmb{\pi}'_n; \pmb{\rho})$, but $\pmb{\pi}\not\in\operatorname{BR}(\pmb{\pi}'; \pmb{\rho})$.
    Consequently, there exists $a\in\mathcal{A}$ that $\sum_{z\in\mathcal{Z}}\rho_{t}(z)\pi_{n, t}(a|s, z) > 0, a\in\mathop{\arg\max}_{a'}\mathcal{Q}_t^{*}(s, a', \mu, \mathcal{I}_t; \pmb{\pi}'_n)$, while $a\not\in\mathop{\arg\max}_{a'}\mathcal{Q}_t^{*}(s, a', \mu, \mathcal{I}_t; \pmb{\pi}')$.
    
    Let $a^\star = \mathop{\arg\max}_{a'}\mathcal{Q}_t^{*}(s, a', \mu, \mathcal{I}_t; \pmb{\pi}_n')$. We denote $\epsilon$ as the gap of action value function.
    \begin{align*}
        \mathcal{Q}_t^{*}(s, a^{\star}, \mu, \mathcal{I}_t; \pmb{\pi}_n') - \mathcal{Q}_t^{*}(s, a, \mu, \mathcal{I}_t; \pmb{\pi}_n') = \epsilon > 0
    \end{align*}
    From the continuity of $\mathcal{Q}_t^{*}(s, a, \mu, \mathcal{I}_t; \pmb{\pi}')=\mathbb{E}_{z\sim\rho_t(\cdot)}Q_{t}^{*}(s, a, \mu, z; \pmb{\pi}')$, there exists $N \in \mathbb{N}$ such that $\vert\mathcal{Q}_t^{*}(s, a, \mu, \mathcal{I}_t; \pmb{\pi}')-\mathcal{Q}_t^{*}(s, a, \mu, \mathcal{I}_t; \pmb{\pi}'_{n})\vert<\frac{\epsilon}{2}$, $\forall n > N, a'\in\mathcal{A}$.

    Then we can induce that
    \begin{align*}
        &\mathcal{Q}_t^{*}(s, a^{\star}, \mu, \mathcal{I}_t; \pmb{\pi}')-\mathcal{Q}_t^{*}(s, a, \mu, \mathcal{I}_t; \pmb{\pi}') \\
        =& \mathcal{Q}_t^{*}(s, a^{\star}, \mu, \mathcal{I}_t; \pmb{\pi}')+\mathcal{Q}_t^{*}(s, a^{\star}, \mu, \mathcal{I}_t; \pmb{\pi}_n')\\
        &-\mathcal{Q}_t^{*}(s, a^{\star}, \mu, \mathcal{I}_t; \pmb{\pi}_n')+\mathcal{Q}_t^{*}(s, a, \mu, \mathcal{I}_t; \pmb{\pi}_n')\\
        &-\mathcal{Q}_t^{*}(s, a, \mu, \mathcal{I}_t; \pmb{\pi}_n')-\mathcal{Q}_t^{*}(s, a, \mu, \mathcal{I}_t; \pmb{\pi}')\\
        \ge&\mathcal{Q}_t^{*}(s, a^{\star}, \mu, \mathcal{I}_t; \pmb{\pi}_n') - \mathcal{Q}_t^{*}(s, a, \mu, \mathcal{I}_t; \pmb{\pi}_n') \\
        &- \vert\mathcal{Q}_t^{*}(s, a^{\star}, \mu, \mathcal{I}_t; \pmb{\pi}')-\mathcal{Q}_t^{*}(s, a^{\star}, \mu, \mathcal{I}_t; \pmb{\pi}_n')\vert\\
        &- \vert\mathcal{Q}_t^{*}(s, a, \mu, \mathcal{I}_t; \pmb{\pi}')-\mathcal{Q}_t^{*}(s, a, \mu, \mathcal{I}_t; \pmb{\pi}_n')\vert\\
        >& \epsilon - \frac{\epsilon}{2} - \frac{\epsilon}{2} = 0,
    \end{align*}
    contradicting with $a\in\mathop{\arg\max}_{a'}\mathcal{Q}_t^{*}(s, a', \mu, \mathcal{I}_t; \pmb{\pi}')$.

    Therefore, $\operatorname{BR}(\pmb{\pi}; \pmb{\rho})$ has a closed graph.
\end{proof}
\begin{lemma}\label{cov}
    $\operatorname{BR}(\pmb{\pi}; \pmb{\rho})$ is a convex set given $\pmb{\pi}$.
\end{lemma}
\begin{proof}
    We assume that $\pmb{\pi}_1\in \operatorname{BR}(\boldsymbol{\pi}'; \pmb{\rho})$, $\pmb{\pi}_2\in \operatorname{BR}(\boldsymbol{\pi}'; \pmb{\rho})$. 
    From \cref{prop}, we have that if $\sum_{z\in\mathcal{Z}}\rho_t(z)\pi_{i, t}(a\mid s, z)>0$, then $a\in \mathop{\arg\max}_{a'\in \mathcal{A}}\mathcal{Q}^{*}(s, a', \mu, I_t; \pmb{\pi}')$, $\forall t \in \mathcal{T}$, $\forall i\in \{1, 2\}$.
    Then the convex combination $\pmb{\pi}=\lambda\pmb{\pi}_1+(1-\lambda)\pmb{\pi}_2, \lambda\in [0, 1]$ also satisfies the requirements of \cref{prop}. Therefore $\pmb{\pi}\in \operatorname{BR}(\pmb{\pi}'; \pmb{\rho})$. $\operatorname{BR}(\pmb{\pi}; \pmb{\rho})$ is a convex set given $\pmb{\pi}$.
\end{proof}
\fixed*
\begin{proof}\label{prffixed}
    As $\pi_t \in \mathcal{P}(\mathcal{A})$, in which $\mathcal{P}(\mathcal{A})$ are simplices with finite dimensions, they are compact. 
    And $\operatorname{BR}(\boldsymbol{\pi}; \pmb{\rho})$ maps to a non-empty set, because the MDP induced by fixed $\boldsymbol{\mu}$ and $\pmb{\rho}$ has an optimal policy. 
    From \cref{closed_graph} and \cref{cov}, the requirements of Kakutani's fixed point theorem holds for $\operatorname{BR}(\boldsymbol{\pi}; \pmb{\rho})$. 
    By Kakutani's fixed point theorem, there exists a fixed point $\boldsymbol{\pi}^* \in \operatorname{BR}(\pmb{\pi}^*; \pmb{\rho})$. And $\forall u\in \mathcal{U}$, $\forall s\in \mathcal{A}$, $\forall t\in \mathcal{T}$,
    \begin{align*}
       \Delta_t(s_t, \mu_t, u; \pmb{\pi}^*, \pmb{\rho}) &=\sum_{z\in\mathcal{Z}}\sum_{a\in\mathcal{A}} \rho_t(z)\pi_t^{*}(a|s, z) \big(Q_t^{\pmb{\pi}^*}(s_t, u(a), \mu_t, z; \pmb{\pi}^*) \\
       &- Q^{\pmb{\pi}^*}(s_t, a, \mu_t, z; \pmb{\pi}^*)\big) \le 0,
    \end{align*}
    where $\mu_{t}=\Phi(\mu_{t-1}, \pi_{t-1}^{*}, z_t)$.
    Then $(\pmb{\pi}^{*}, \pmb{\rho})$ is an AMFCE.
\end{proof}
\subsection{Proof of \cref{corr:relation}}
\relation*
\begin{proof}\label{prfrelation}
    If $(\pmb{\pi}, \pmb{\mu})$ represents a MFNE, the following condition holds \cite{cui2021approximately}:
    $\pi_t(a\mid s, z) > 0$ is a sufficient condition for $$a \in \mathop{\arg\max}_{a' \in \mathcal A} Q_t^{*}(s, a', \mu, z; \pmb{\pi})$$.
    
    If the correlation device $\pmb{\rho}=\{\rho_t\}_{t\in\mathcal{T}}$ satisfies $\rho_t=\delta(z)$ for all $t\in\mathcal{T}$,
    $\sum_{z}\rho_{t}(z)\pi_{t}(a\mid s, z) > 0$ is a sufficient condition for $a \in \mathop{\arg\max}_{a' \in \mathcal A} \mathbb{E}_{z\sim\rho_{t}^{\rm pred}(\cdot|\mathcal{I}_{t})}\left[Q_{t}^{*}(s, a', \mu, z; \pmb{\pi})\right]$.
    
    Additionally, the mean field flow $\pmb{\mu}$ satisfies $\mu_{t}=\Phi(\mu_{t-1}, \pi_{t-1}, z)$.
    Therefore, $(\pmb{\pi}, \pmb{\rho})$ forms an Adaptive Mean Field Correlated Equilibrium (AMFCE).
    \end{proof}
\subsection{Proof of Proposition \ref{tstep}}
\tstep*
\begin{proof}\label{prftstep}
\textbf{(Sufficient Condition)}. Suppose that $(\pmb{\pi}, \pmb{\rho})$ is a solution of AMFCE but the inequality in Proposition \ref{tstep} does not hold.
There exists some $t$ and trajectory such that 
\begin{align*}
    \mathbb{E}_{\pmb{\pi}, \pmb{\pi}, \pmb{\rho}}\left[\sum_{t=0}^{T} \gamma^t r(s_t, a_t, \mu_t)\middle|a_{0:T}\right] > J(\pmb{\pi}, \pmb{\pi}, \pmb{\rho})
\end{align*}
From the definition of AMFCE, 
\begin{align*}
    \sum_{a\in\mathcal{A}}\sum_{z\in\mathcal{Z}}\rho_t(z)\pi_t(a|s, z)\Big[Q_t^{\pmb{\pi}}(s, a, \mu_t, z; \pmb{\pi})-Q_t^{\pmb{\pi}}(s, a', \mu_t, z; \pmb{\pi})\Big]\ge0
\end{align*}
We have that
\begin{align*}
    &\mathbb{E}_{\pmb{\pi}, \pmb{\pi}, \pmb{\rho}}\left[\sum_{t=0}^{T} \gamma^t r(s_t, a_t, \mu_t)\middle|a_{0:T}\right]\\
    =&\mathbb{E}_{\pmb{\pi}, \pmb{\pi}, \pmb{\rho}}\left[\sum_{t=0}^{T-1} \gamma^t r(a_t, s_t, \mu_t) + \gamma^{T}r(s_{T}, a_{T}, \mu_{T})\middle|a_{0:T}\right]\\
    \le&\mathbb{E}_{\pmb{\pi}, \pmb{\pi}, \pmb{\rho}}\left[\sum_{t=0}^{T-1} \gamma^t r(a_t, s_t, \mu_t) + \gamma^{T}\mathbb{E}\left[Q_T^{\pmb{\pi}}(s_T, a, \mu_T, z; \pmb{\pi})\right]\middle|a_{0:T-1}\right]
\end{align*}
The inner expectation is taken with respect to $z\sim\rho_{T}(\cdot)$, $a\sim\pi_T(\cdot|s_T, z)$.
Similarly, we can induce that
\begin{align*}
    &\mathbb{E}_{\pmb{\pi}, \pmb{\pi}, \pmb{\rho}}\left[\sum_{t=0}^{T} \gamma^t r(s_t, a_t, \mu_t)\middle|a_{0:T}\right]\\
    \le&\mathbb{E}_{\pmb{\pi}, \pmb{\pi}, \pmb{\rho}}\Bigg[\sum_{t=0}^{T-2} \gamma^t r(a_t, s_t, \mu_t) + \gamma^{T-1} r(s_{T-1}, a_{T-1}, \mu_{T-1}) \\
    &+ \gamma^{T}\mathbb{E}\left[Q_T^{\pmb{\pi}}(s_T, a, \mu_T, z; \pmb{\pi})\right]\bigg|a_{0:T-1}\Bigg]\\
    \le&\mathbb{E}_{\pmb{\pi}, \pmb{\pi}, \pmb{\rho}}\left[\sum_{t=0}^{T-2} \gamma^t r(a_t, s_t, \mu_t) + \gamma^{T-1}\mathbb{E}[Q_{T-1}^{\pmb{\pi}}(s_{T-1}, a, \mu_{T-1}, z;\pmb{\pi})]\middle|a_{0:T-2}\right]\\
    \le&\mathbb{E}_{\pmb{\pi}, \pmb{\pi}, \pmb{\rho}}\bigg[Q_0^{\pmb{\pi}}(s_{0}, a, \mu_0, z; \pmb{\pi})\bigg]=J(\pmb{\pi}, \pmb{\pi}, \pmb{\rho}),
\end{align*}
where the last expectation is taken with respect to $z\sim\rho_0, s_0\sim\mu_0(\cdot), a\sim\pi_{0}(\cdot|s_0, z)$.

It contradicts with the assumption.

\textbf{(Necessary Condition)}. We assume that the inequality holds and $(\pmb{\pi}, \pmb{\rho})$ is not an AMFCE. There exists a time step $t\in \mathcal{T}$ such that $\Delta_{t}(s, \mu, u; \pmb{\pi}, \pmb{\rho})=\mathbb{E}[Q_{t}^{\pmb{\pi}}(s, u(a), \mu, z)-Q_{t}^{\pmb{\pi}}(s, a, \mu, z)]>0$. Then agent can achieve a strictly higher expected return if it chooses action $u(a)$ when it is recommended action $a$ at time step $t$. It implies that there exists an action sequence such that $\mathcal{R}(a_{0:T}, \pmb{\pi}, \pmb{\rho})>0$, which conflicts with the assumption.
\end{proof}
\subsection{Proof of Theorem \ref{thm:dual}}
\dual*
\begin{proof}\label{prfdual}
    We note that
    \begin{align*}
        \sum_{\tau_k\in\mathcal{D}_E}\lambda_{\pmb{\pi}^*}(\tau_i)&\mathbb{E}_{\pmb{\pi}, \pmb{\pi}, \pmb{\rho}}\left[\sum_{t=0}^{T}\gamma^t r(s_t, a_t, \mu_t)\middle|a_{0:T}\right]\\
        =&\mathbb{E}_{\pmb{\pi}^{*}}\left[\mathbb{E}_{\pmb{\pi}, \pmb{\pi}, \pmb{\rho}}\left[\sum_{t=0}^{T}\gamma^t r(s_t, a_t, \mu_t)\middle|a_{0:T}\right]\right]\\
        =&\mathbb{E}_{\pmb{\pi}^*, \pmb{\pi}, \pmb{\rho}}\left[\sum_{t=0}^{T}\gamma^{t}r(s_t, a_t, \mu_t)\right]
    \end{align*}
    The $\mathbb{E}_{\pmb{\pi}^{*}}$ is taken with respect to $a_t\sim\pi_t^{*}(\cdot|s_t, z_t)$. 
    Then we can derive the conclusion directly.
    \begin{align*}
        L(\boldsymbol{\pi}, \pmb{\rho}, \lambda_{\pmb{\pi}^*}, r)=&J(\pmb{\pi}^*, \pmb{\pi}, \pmb{\rho})-J(\pmb{\pi}, \pmb{\pi}, \pmb{\rho})
    \end{align*}
\end{proof}
\subsection{Proof of Proposition \ref{GAIL}}
\GAIL*
\begin{proof}\label{prfGAIL}
    We denote $\tilde{r}=\operatorname{AMFCE-IRL}(\boldsymbol{\pi}^E, \pmb{\rho}^E)$.
    The saddle point of $L(\pmb{\pi}, \pmb{\rho}, \lambda, r)$ is $\lambda_{\pmb{\pi}}^{E}(\tau_k)=\prod_{t=0}^{T}\pi_t^E(a_t|s_t, z_t)$ and $\tilde{r}$, where $(\pmb{\pi}^E, \pmb{\rho}^E)\in\operatorname{AMFCE}$.
    So given expert demonstrations sampled from 
    $(\pmb{\pi}^E, \pmb{\rho}^E)$
    , we can recover $\pmb{\pi}^E$ by \cref{recover}.
    \begin{align}\label{recover}
        \pmb{\pi}&=
        \mathop{\arg\min}_{\pmb{\pi}}J(\pmb{\pi}^E, \pmb{\pi}^E, \pmb{\rho}^E) - \mathbb{E}_{\pmb{\pi}^E, \pmb{\pi}^E, \pmb{\rho}^E}[\sum_{t=0}^{T}\gamma^{t}\tilde{r}(s_t, a_t, \mu_t)]\nonumber\\
        &=\mathop{\arg\min}_{\pmb{\pi}} \max_{r} J(\pmb{\pi}^E, \pmb{\pi}^E, \pmb{\rho}^E) - J(\pmb{\pi}, \pmb{\pi}^E, \pmb{\rho}^E) - \psi_{GA}(r)
    \end{align}
    If we select $\psi_{GA}$ as the regularizer, and 
    make the change of variables $r(s, a, \mu) = \log\big(D_{\omega}(s, a, \mu)\big)$, we get
    \begin{align*}
        &\max_{r} J(\pmb{\pi}^E, \pmb{\pi}^E, \pmb{\rho}^E) - J(\pmb{\pi}, \pmb{\pi}^E, \pmb{\rho}^E) - \psi_{GA}(r)\\
        =&\max_{\omega}\mathbb{E}_{\pmb{\pi}_E, \pmb{\pi}_E, \pmb{\rho}_E}\left[\sum_{t=0}^{T}\gamma^{t}\log(D_{\omega}(s, a, \mu))\right] \\
        &- \mathbb{E}_{\pmb{\pi}, \pmb{\pi}^E, \pmb{\rho}_E}\left[\sum_{t=0}^{T}\gamma^{t}\log(D_{\omega}(s, a, \mu))\right] \\
        &- \max_{\omega}\mathbb{E}_{\pmb{\pi}, \pmb{\pi}_E, \pmb{\rho}_E}\left[\sum_{t=0}^T g(r(s_t, a_t, \mu_t))\right]\\
        =&\max_{\omega}\mathbb{E}_{\pmb{\pi}, \pmb{\pi}^E, \pmb{\rho}^E}\left[\sum_{t=0}^{T}\gamma^{t}\log D_\omega(s_t, a_t, \mu_t)\right] \\
        &+ \mathbb{E}_{\pmb{\pi}^E, \pmb{\pi}^E, \pmb{\rho}^E}\left[\sum_{t=0}^{T}\gamma^{t}\log\big(1 - D_\omega(s_t, a_t, \mu_t)\big)\right].
    \end{align*}

\end{proof}
\subsection{Proof of Proposition \ref{gradofrho}}
\gradofrho*
\begin{proof}\label{prfgradofrho}
    The gradient of parameterized $\rho^\phi$ is 
    \begin{align*}
        \nabla_\phi\Delta_t&(s, \mu, u; \pmb{\pi}, \pmb{\rho})=\nabla_{\phi}\sum_{z\in\mathcal{Z}}\rho_t^{\phi}(z)\sum_{a\in\mathcal{A}}\pi_t(a|s, z)Q_t^{\pmb{\pi}}(s, a, \mu, z; \pmb{\pi})\\
        =&\sum_{z\in\mathcal{Z}}\nabla_{\phi}\rho_{t}^{\phi}(z)\sum_{a\in\mathcal{A}}\pi_{t}(a|s, z)Q_t^{\pmb{\pi}}(s, a, \mu, z; \pmb{\pi})\\
        =&\mathbb{E}_{z\sim\rho_{t}^{\phi}(\cdot)}\Big[\sum_{a\in\mathcal{A}}\pi_t(a|s, z)Q_t^{\pmb{\pi}}(s, a, \mu, z; \pmb{\pi})\nabla_{\phi}\log\rho_t^{\phi}(z)\Big]\\
        =&\mathbb{E}_{z\sim\rho_{t}^{\phi}(\cdot)}\bigg[\nabla_{\phi}\log\rho_t^{\phi}(z)\mathbb{E}_{a\sim\pi_t(\cdot|s, z)}Q_{t}^{\pmb{\pi}}(s, a, \mu, z; \pmb{\pi})\bigg].
    \end{align*}
\end{proof}
\subsection{Proof of the Theorem \ref{thm: bound}}
\label{prf: bound}
\bound*
\begin{proof}
    When the discriminator achieves its optimum
    \begin{align}
        D_\omega^*(s_t, a_t, \hat{\mu}_t)=\frac{2\eta_t^{\pmb{\pi}}(s_t, a_t, \hat{\mu}_t)}{\eta_t^{\pmb{\pi}}(s_t, a_t, \hat{\mu}_t)+\eta_t^{E}(s_t, a_t, \hat{\mu}_t)},
    \end{align}
    We denote $\hat{\mu}_t=\mathrm{Sig}(\pmb{z}_{0:t})$.
    we can derive that MFCIL is to minimize the state-action distribution discrepancy between the expert policy and the recovered policy with the Jensen-Shannon (JS) divergence (up to a constant):
    \begin{smalleralign}
        D_{\mathrm{JS}}(\eta_t^{E}(s, a, \hat{\mu}),\eta_t^{\pi}(s, a, \hat{\mu}))&\triangleq\frac12\bigg[D_{\mathrm{KL}}\left(\eta_t^{\pi}(s, a, \hat{\mu}),\frac{\eta_t^{\pi}(s, a, \hat{\mu})+\eta_t^{E}(s, a, \hat{\mu})}2\right)\\
        &\quad+D_{\mathrm{KL}}\left(\eta_t^{E}(s, a, \hat{\mu}),\frac{\eta_t^{\pi}+\eta_t^{E}}2\right)\bigg],
    \end{smalleralign}
    where $\eta_t^{\pi}(s, a, \hat{\mu})$ and $\eta_t^{E}(s, a, \hat{\mu})$ is the occupancy measure of the recovered policy at time step $t$.
    We define the occupancy measure of the expert policy as $\eta_t^{E}$ and the state distribution of agents following the recovered policy as $\eta_t^{\pi}$.
    \begin{align}
        \begin{cases}
            \eta_t^{\pi}(s, a, \hat{\mu})=\rho_t^E(z_t)\pi_t(a|s, z_t)\eta_t(s, \hat{\mu}_t)\\
            \eta_t^{\pi}(s, \hat{\mu})=\sum_{a'\in\mathcal{A}}\sum_{s'\in\mathcal{S}}\eta_{t-1}^\pi(s', a')P(s|s', a', \mu_{t-1})\\
            \eta_0^{\pi}(s, \hat{\mu})=\mu_0
        \end{cases}
    \end{align}
    \begin{align}
        \begin{cases}
            \eta_t^{E}(s, a, \hat{\mu})=\rho_t^E(z_t)\pi_t^E(a|s, z_t)\eta_t(s, \hat{\mu}_t)\\
            \eta_t^{E}(s, \hat{\mu})=\sum_{a'\in\mathcal{A}}\sum_{s'\in\mathcal{S}}\eta_{t-1}^E(s', a')P(s|s', a', \mu_{t-1})\\
            \eta_0^{E}(s, \hat{\mu})=\mu_0
        \end{cases}
    \end{align}
    Here, $\hat{\mu}_t=\mathrm{Sig}(\pmb{z}_{0:t})$.
    Under the Assumption \ref{assp:bound},
    \begin{align}
        \max_{\omega}&\mathbb{E}_{\pmb{\pi}, \pmb{\pi}, \pmb{\rho}^E}\bigg[\sum_{t=0}^{T}\gamma^t\log D_\omega^*(s_t, a_t, \hat{\mu}_t)\bigg]\nonumber\\
        &+\mathbb{E}_{\pmb{\pi}^E, \pmb{\pi}^E, \pmb{\rho}^E}\bigg[\sum_{t=0}^{T}\gamma^t\log\big(1 - D_\omega^*(s_t, a_t, \hat{\mu}_t)\big)\bigg]\le\epsilon,
    \end{align}
    we can derive that
    \begin{align}
        \sum_{t=0}^{T}\gamma^tD_{\mathrm{JS}}\left(\eta_t^{E}(s, a, \hat{\mu}), \eta_t^{\pi}(s, a, \hat{\mu})\right)\le\epsilon.
    \end{align}
    From Pinsker's inequality, we have
    \begin{smalleralign}
        \frac{1}{2}&\left\|\eta_t^E(s, a, \hat{\mu})-\eta_t^\pi(s, a, \hat{\mu})\right\|_1\le\nonumber\\&\sqrt{2D_{\mathrm{KL}}\left(\eta_t^E(s, a, \hat{\mu}), \frac{\eta_t^E(s, a, \hat{\mu})+\eta_t^\pi(s, a, \hat{\mu})}{2}\right)}
    \end{smalleralign}
    and
    \begin{align*}
        \frac{1}{2}&\left\|\eta_t^E(s, a, \hat{\mu})-\eta_t^\pi(s, a, \hat{\mu})\right\|_1\le\nonumber\\&\sqrt{2D_{\mathrm{KL}}\left(\eta_t^\pi(s, a, \hat{\mu}), \frac{\eta_t^E(s, a, \hat{\mu})+\eta_t^\pi(s, a, \hat{\mu})}{2}\right)}.
    \end{align*}
    From the Jensen inequality, we have that
    \begin{align*}
        \left\|\eta_t^E(s, a, \hat{\mu})-\eta_t^\pi(s, a, \hat{\mu})\right\|_1\le2\sqrt{2D_{\mathrm{JS}}(\eta_t^{E}(s, a, \hat{\mu}), \eta_t^{\pi}(s, a, \hat{\mu}))}
    \end{align*}
    We use again the Jensen inequality
    \begin{align*}
        \frac{1}{T}&\sum_{t=0}^{T}\gamma^t\left\|\eta_t^E(s, a, \hat{\mu})-\eta_t^\pi(s, a, \hat{\mu})\right\|_1\\&\le\frac{1}{T}\sum_{t=0}^{T}2\sqrt{2\gamma^{2t}D_{\mathrm{JS}}(\eta_t^{E}(s, a, \hat{\mu}), \eta_t^{\pi}(s, a, \hat{\mu}))}\\
        &\le2\sqrt{\frac{1}{T}\sum_{t=0}^{T}2\gamma^{t}D_{\mathrm{JS}}(\eta_t^{E}(s, a, \hat{\mu}), \eta_t^{\pi}(s, a, \hat{\mu}))}\\
        &\le2\sqrt{\frac{2\epsilon}{T}}
    \end{align*}
    Therefore, we have
    \begin{align*}
        \sum_{t=0}^{T}\gamma^t\left\|\eta_t^E(s, a, \hat{\mu})-\eta_t^\pi(s, a, \hat{\mu})\right\|_1\le2\sqrt{2\epsilon T}.
    \end{align*}
    We then bound the Jensen-Shannon divergence of state occupancy.
    From Jensen inequality, we have that
    \begin{align*}  
        D_{\mathrm{KL}}&\left(\eta_t^E(s, a, \hat{\mu}), \frac{\eta_t^\pi(s, a, \hat{\mu})+\eta_t^E(s, a, \hat{\mu})}{2}\right)\\
        &=\mathbb{E}_{\eta_t^E(s, a, \hat{\mu})}\log\frac{2\eta_t^E(s, a, \hat{\mu})}{\eta_t^\pi(s, a, \hat{\mu})+\eta_t^E(s, a, \hat{\mu})}\\
        &=\sum_{s, a, \hat{\mu}}\eta_t^E(s, a, \hat{\mu})\log\frac{2\eta_t^E(s, a, \hat{\mu})}{\eta_t^\pi(s, a, \hat{\mu})+\eta_t^E(s, a, \hat{\mu})}\\
        &=\sum_{s\in\mathcal{S}}\eta_t^E(s, \hat{\mu})\sum_{a\in\mathcal{A}}\frac{\eta_t^E(s, a, \hat{\mu})}{\eta_t^E(s, \hat{\mu})}\log\frac{2\eta_t^E(s, a, \hat{\mu})}{\eta_t^\pi(s, a, \hat{\mu})+\eta_t^E(s, a, \hat{\mu})}\\
        &\ge\sum_{s\in\mathcal{S}}\eta_t^E(s, \hat{\mu})\log\frac{2\eta_t^E(s, \hat{\mu})}{\sum_{a\in\mathcal{A}}\eta_t^\pi(s, a, \hat{\mu})+\eta_t^E(s, \hat{\mu})}\\
        &=\sum_{s\in\mathcal{S}}\eta_t^E(s, \hat{\mu})\log\frac{2\eta_t^E(s, \hat{\mu})}{\eta_t^\pi(s, \hat{\mu})+\eta_t^E(s, \hat{\mu})}\\
        &=D_{\mathrm{KL}}\left(\eta_t^E(s, \hat{\mu}), \frac{\eta_t^\pi(s, \hat{\mu})+\eta_t^E(s, \hat{\mu})}{2}\right)
    \end{align*}
    Similarly, we have
    \begin{align*}
        D_{\mathrm{KL}}&\left(\eta_t^\pi(s, a, \hat{\mu}), \frac{\eta_t^\pi(s, a, \hat{\mu})+\eta_t^E(s, a, \hat{\mu})}{2}\right)\nonumber\\\ge& D_{\mathrm{KL}}\left(\eta_t^\pi(s, \hat{\mu}), \frac{\eta_t^\pi(s, \hat{\mu})+\eta_t^E(s, \hat{\mu})}{2}\right).
    \end{align*}
    Therefore, the Jensen-Shannon divergence of state occupancy is bounded by
    \begin{align*}
        \sum_{t=0}^{T}\gamma^tD_{\mathrm{JS}}(\eta_t^E(s, \hat{\mu}), \eta_t^\pi(s, \hat{\mu}))\le \sum_{t=0}^{T}\gamma^tD_{\mathrm{JS}}(\eta_t^E(s, a, \hat{\mu}), \eta_t^\pi(s, a, \hat{\mu}))\le\epsilon.
    \end{align*}
    Similarly, we can derive that
    \begin{align*}
        \sum_{t=0}^{T}\gamma^t\left\|\eta_t^E(s, \hat{\mu})-\eta_t^\pi(s, \hat{\mu})\right\|_1\le2\sqrt{2\epsilon T}.
    \end{align*}
    We define the $\pmb{\mu}'$ as $\mu_t'=\Phi(\mu_{t-1}', \pi_{t-1}^E, z_{t-1})$.
    Therefore, we have
    \begin{align*}
        &\left\|J(\pmb{\pi}, \pmb{\pi}, \pmb{\rho}^E)-J(\pmb{\pi}^E, \pmb{\pi}^E, \pmb{\rho}^E)\right\|_1=\\ &\left\|\mathbb{E}_{\pmb{\pi}, \pmb{\pi}, \pmb{\rho}^E}\left[\sum_{t=0}^T\gamma^t r(s_t, a_t, \mu_t)\right]-\mathbb{E}_{\pmb{\pi}^E, \pmb{\pi}^E, \pmb{\rho}^E}\left[\sum_{t=0}^T\gamma^t r(s_t, a_t, \mu_t')\right]\right\|_1\\
        \le &\left\|\mathbb{E}_{\pmb{\pi}, \pmb{\pi}, \pmb{\rho}^E}\left[\sum_{t=0}^T\gamma^t \left(r(s_t, a_t, \mu_t)-r(s_t, a_t, \mu_t')\right)\right]\right\|_1+\\
        &\left\|\mathbb{E}_{\pmb{\pi}, \pmb{\pi}, \pmb{\rho}^E}\left[\sum_{t=0}^T\gamma^t r(s_t, a_t, \mu_t')\right]-\mathbb{E}_{\pmb{\pi}^E, \pmb{\pi}^E, \pmb{\rho}^E}\left[\sum_{t=0}^T\gamma^t r(s_t, a_t, \mu_t')\right]\right\|_1\\
        \le &2L_R\sqrt{2\epsilon T}+2r_{\max}\sqrt{2\epsilon T}\\
        \le &2(L_R+r_{\max})\sqrt{2\epsilon T}.
    \end{align*}
    From \cref{lemma: simulation}, we have
    \begin{align*}
        &\left\|\mathcal{R}(a_{0:T}, \pmb{\pi}, \pmb{\rho}^E)-\mathcal{R}(a_{0:T}, \pmb{\pi}^E, \pmb{\rho}^E)\right\|_1\\&\le\left\|\mathbb{E}_{\pmb{\pi}, \pmb{\pi}, \pmb{\rho}^E}\Big[\sum_{t=0}^{T} \gamma^t r(s_t, a_t, \mu_{t})\Big|a_{0:T}\Big]\right.\\
        &\quad\left.-\mathbb{E}_{\pmb{\pi}^E, \pmb{\pi}^E, \pmb{\rho}^E}\Big[\sum_{t=0}^{T} \gamma^t r(s_t, a_t, \mu_{t})\Big|a_{0:T}\Big]\right\|_1\\
        &\quad+\left\|J(\pmb{\pi}, \pmb{\pi}, \pmb{\rho}^E)-J(\pmb{\pi}^E, \pmb{\pi}^E, \pmb{\rho}^E)\right\|_1\\
        &\le 2\left(2L_R+r_{\max}+\gamma TL_Pr_{\max}\right)\sqrt{2\epsilon T}.
    \end{align*}
    Since $\mathcal{R}(a_{0:T}, \pmb{\pi}^E, \pmb{\rho}^E)\le0$, we have
    \begin{align*}
        \mathcal{R}(a_{0:T}, \pmb{\pi}, \pmb{\rho}^E)\le&\left\|\mathcal{R}(a_{0:T}, \pmb{\pi}, \pmb{\rho}^E)-\mathcal{R}(a_{0:T}, \pmb{\pi}^E, \pmb{\rho}^E)\right\|_1\\\le& 2\left(2L_R+r_{\max}+\gamma TL_Pr_{\max}\right)\sqrt{2\epsilon T}.
    \end{align*}
\end{proof}
\begin{lemma}\label{lemma: simulation}
    \begin{align*}
        \Bigg\|\mathbb{E}_{\pmb{\pi}, \pmb{\pi}, \pmb{\rho}^E}&\Big[\sum_{t=k}^{T} \gamma^t r(s_t, a_t, \mu_{t})\Big|a_{k:T}\Big]\\-&\mathbb{E}_{\pmb{\pi}^E, \pmb{\pi}^E, \pmb{\rho}^E}\Big[\sum_{t=k}^{T} \gamma^t r(s_t, a_t, \mu_{t})\Big|a_{k:T}\Big]\Bigg\|_1\\\le& (L_R+\gamma TL_Pr_{\max})\sum_{t=k}^T\gamma^{t-k}\|\eta_t^\pi-\eta_t^{E}\|_1.
    \end{align*}
\end{lemma}
\begin{proof}
    At the step $k = T$, this is clearly true since the two value functions only differ in the reward at the final step.
    For the inductive step, we have
    \begin{align*}
        &\bigg\|\mathbb{E}_{\pmb{\pi}, \pmb{\pi}, \pmb{\rho}^E}\Big[\sum_{t=k}^{T} \gamma^t r(s_t, a_t, \mu_{t})\Big|a_{k:T}\Big]\\
        &\quad\quad\quad-\mathbb{E}_{\pmb{\pi}^E, \pmb{\pi}^E, \pmb{\rho}^E}\Big[\sum_{t=k}^{T} \gamma^t r(s_t, a_t, \mu_{t})\Big|a_{k:T}\Big]\bigg\|_1\\
        =&\Bigg\|\mathbb{E}_{\pmb{\pi}, \pmb{\pi}, \pmb{\rho}^E}\Big[r(s_k, a_k, \mu_{k})\\
        &+\gamma\mathbb{E}_{\pmb{\pi}, \pmb{\pi}, \pmb{\rho}^E}\Big[\sum_{t=k+1}^{T} \gamma^{t-k} r(s_t, a_t, \mu_{t})\Big|a_{k+1:T}\Big]\Big|a_{k:T}\Big]\\
        &-\mathbb{E}_{\pmb{\pi}^E, \pmb{\pi}^E, \pmb{\rho}^E}\Big[r(s_k, a_k, \mu_{k})\\
        &+\gamma\mathbb{E}_{\pmb{\pi}^E, \pmb{\pi}^E, \pmb{\rho}^E}\Big[\sum_{t=k+1}^{T} \gamma^{t-k} r(s_t, a_t, \mu_{t})\Big|a_{k+1:T}\Big]\Big|a_{k:T}\Big]\Bigg\|_1\\
        \le&\left\|r(s_k, a_k, \eta_{t}^E)-r(s_k, a_k, \eta_{t}^\pi)\right\|_1\\
        &+\left\|P(s_{t+1}|s_t, a_t, \eta_{t}^E)-P(s_{t+1}|s_t, a_t, \eta_{t}^\pi)\right\|_1\\
        &\left\|\mathbb{E}_{\pmb{\pi}^E, \pmb{\pi}^E, \pmb{\rho}^E}\Big[\sum_{t=k+1}^{T} \gamma^{t-k} r(s_t, a_t, \mu_{t})\Big|a_{k+1:T}\Big]\right\|_1\\
        &+\gamma\Bigg\|\mathbb{E}_{\pmb{\pi}, \pmb{\pi}, \pmb{\rho}^E}\Big[\sum_{t=k+1}^{T} \gamma^{t-k} r(s_t, a_t, \mu_{t})\Big|a_{k+1:T}\Big]\\
        &-\mathbb{E}_{\pmb{\pi}^E, \pmb{\pi}^E, \pmb{\rho}^E}\Big[\sum_{t=k+1}^{T} \gamma^{t-k} r(s_t, a_t, \mu_{t})\Big|a_{k+1:T}\Big]\Bigg\|_1\\
        \le& L_R\|\eta_k^E-\eta_k^\pi\|_1+\gamma Tr_{\max}L_P\|\eta_k^E-\eta_k^\pi\|_1\\
        &\quad+\gamma\Bigg\|\mathbb{E}_{\pmb{\pi}, \pmb{\pi}, \pmb{\rho}^E}\Big[\sum_{t=k+1}^{T} \gamma^{t-k} r(s_t, a_t, \mu_{t})\Big|a_{k+1:T}\Big]\\
        &\quad-\mathbb{E}_{\pmb{\pi}^E, \pmb{\pi}^E, \pmb{\rho}^E}\Big[\sum_{t=k+1}^{T} \gamma^{t-k} r(s_t, a_t, \mu_{t})\Big|a_{k+1:T}\Big]\Bigg\|_1\\
        \le& L_R\|\eta_k^E-\eta_k^\pi\|_1\\
        &+\gamma Tr_{\max}L_P\|\eta_k^E-\eta_k^\pi\|_k+(L_R+\gamma L_PTr_{\max})\sum_{t=k+1}^T\gamma^{t-k}\|\eta_t^\pi-\eta_t^{E}\|_1\\
        \le&(L_R+\gamma TL_Pr_{\max})\sum_{t=k}^T\gamma^{t-k}\|\eta_t^\pi-\eta_t^{E}\|_1
    \end{align*}
\end{proof}
\subsection{Proof of the \cref{corr: imitation-gap}}\label{prf: imitation-gap}
{\imitationgap*
\begin{proof}
    We denote the optimal policy $\pmb{\pi}' = \arg\max_{\hat{\pmb{\pi}}}J(\hat{\pmb{\pi}}, \pmb{\pi}, \pmb{\rho}^E)$ when the population follows the recovered policy $\pmb{\pi}$ and correlation device $\pmb{\rho}^E$.
    
    We have
    \begin{align*}
        &J(\pmb{\pi}', \pmb{\pi}, \pmb{\rho}^E) - J(\pmb{\pi}, \pmb{\pi}, \pmb{\rho}^E)\\&\le \left\|\mathbb{E}_{\pmb{\pi}', \pmb{\pi}, \pmb{\rho}^E}\bigg[\sum_{t=0}^{T}\gamma^tr(s_t, a_t, \mu_t)\bigg]-\mathbb{E}_{\pmb{\pi}', \pmb{\pi}, \pmb{\rho}^E}\bigg[\sum_{t=0}^{T}\gamma^tr(s_t, a_t, \mu_t')\bigg]\right\|_1\\
        &\quad+\left\|\mathbb{E}_{\pmb{\pi}', \pmb{\pi}, \pmb{\rho}^E}\bigg[\sum_{t=0}^{T}\gamma^tr(s_t, a_t, \mu_t')\bigg]-\mathbb{E}_{\pmb{\pi}', \pmb{\pi}^E, \pmb{\rho}^E}\bigg[\sum_{t=0}^{T}\gamma^tr(s_t, a_t, \mu_t')\bigg]\right\|_1\\
        &\quad+\mathbb{E}_{\pmb{\pi}', \pmb{\pi}^E, \pmb{\rho}^E}\bigg[\sum_{t=0}^{T}\gamma^tr(s_t, a_t, \mu_t')\bigg]-J(\pmb{\pi}^E, \pmb{\pi}^E, \pmb{\rho}^E)\\
        &\quad+J(\pmb{\pi}^E, \pmb{\pi}^E, \pmb{\rho}^E)-J(\pmb{\pi}, \pmb{\pi}, \pmb{\rho}^E)\\
        &\le 2L_R\sqrt{2\epsilon T}+2(L_R+\gamma TL_Pr_{\max})\sqrt{2\epsilon T}+0+2(L_R+r_{\max})\sqrt{2\epsilon T}\\
        &=2(3L_R+\gamma TL_Pr_{\max}+r_{\max})\sqrt{2\epsilon T},
    \end{align*}
    where $\mu_t'=\Phi(\mu_{t-1}', \pi_{t-1}^E, z_{t-1})$.
\end{proof}}
\section{The camparison between AMFCE and common noise}\label{sec:common_noise}
In this subsection, we compare the AMFCE with the common noise equilibrium.
In the context of MFG with common noise, the optimal policy aims to maximize the expected return under the {\bf prior} distribution of common noise, such that
\begin{equation*}
    Q_{n}^{\pi,\mu}(x,u(a)|\Xi_{n})-Q_{n}^{\pi,\mu}(x,a|\Xi_{n})\leq 0,
\end{equation*}
where $u(a)\in\mathcal{A}$ is the modified action and the Q function is defined as following
\begin{equation*}
    \begin{aligned}
        Q_N^{\pi,\mu}(x,a|\Xi_N)=&r(x,a,\mu_{N|\Xi_N},\xi_N),    
    \end{aligned}
\end{equation*}
\begin{equation*}
    \begin{aligned}
        &Q_{n-1}^{\pi,\mu}(x,a|\Xi_{n-1})=\underbrace{\sum_{\xi}P(\xi_{n-1}=\xi|\Xi_{n-1})}\Big[ r(x,a,\mu_{n-1,\Xi_{n-1}},\xi)\\&+\sum_{x^{\prime}\in\mathcal{X}}p(x^{\prime}|x,a,\xi)E_{\pmb{b}\sim\pi_n(\cdot|x^{\prime},\Xi_{n-1}\cdot\xi)}\left[Q_n^{\pi,\mu}(x^{\prime},b|\Xi_{n-1}.\xi)\right]\Big],
    \end{aligned}
\end{equation*}

The equilibrium of common noise is under the framework of Nash equilibrium.

In contrast, the AMFCE framework aims to maximize the expectation under the {\bf posterior} distribution of correlated signal $z$ of the Q-function corresponding to the recommended action $a$, as expressed by:

$$\underbrace{\sum_{z}\frac{\rho_t(z)\pi_t(a|s, z)}{\sum_{a}\rho_t(z)\pi_t(a|s, z)}}[Q_t^{\pmb{\pi}}( s, u(a) , \mu, z; \pmb{\pi})-Q_t^{\pmb{\pi}}(s,a,\mu,z;\pmb{\pi})]\leq0.$$

To illustrate the difference between AMFCE and MFNE with common noise, consider a mean field game $\mathcal{G}$. In $\mathcal{G}$, the state space $S=\{C,L,R\}$, and the action space $A=\{L,R\}$. The initial mean field $\mu_0(C)=1$, and the reward function is defined as $r(s,a,\mu) = 1_{\{ s= L\} }\mu( L) + 1_{\{ s= R\} }\mu( R)$. The environment dynamics are deterministic: $P(s_{t+1}=R|s_t=\cdot,a=R)=1$ and $P(s_{t+1}=L|s_t=\cdot,a=L)=1$. Correlated signals are sampled from the space $\mathcal{Z}=\{0, 1\}$ with equal probability $\rho(z=0)=\rho(z=1)=0.5$.
In this scenario, the policy $\pi(a=L|s=\cdot, z=0)=\frac23$ and $\pi(a=L|s=\cdot, z=1)=\frac13$ constitute an AMFCE but not a MFNE with common noise.
Specifically, the policies $\pi(a=L|s=\cdot, z=0)=1, 0, \frac12$ and $\pi(a=L|s=\cdot, z=1)=1, 0, \frac12$ constitute MFNE with common noise, while all of them are also AMFCE.

\section{Experiment detail}\label{Detail}
The experiments were run on the server with AMD EPYC 7742 64-Core Processor and NVIDIA A100 40GB.

Due to the instability nature of generative adversarial networks (GANs) \cite{DBLP:conf/iclr/ArjovskyB17, DBLP:conf/icml/MeschederGN18}, the convergence of Algorithm \ref{algo} may not be not guaranteed. To address this issue, we integrated the gradient penalty into the objective function of MFCIL to stabilize the training of policy $\boldsymbol{\pi}$. It has been proven that GAN training with zero-centered will enhance the training stability \cite{DBLP:conf/icml/MeschederGN18}.
To provide a fair comparison, we used Actor-Critic (AC) algorithm for both MFCIL, MFAIRL, and MFIRL.
The input of AC is an extended state, a concatenation of state, action, time step, and signature.
The input of the discriminator is the extended state and the action.
We did not use signature in the Sequential Squeeze with $\mathcal{T}=\{0, 1\}$ and RPS because signature requires the length of sequential data is larger than 1.
For games with the sequential setting, the depth of truncated signature is 3.
For actor and critic networks of AC, we adopt two-layer perceptrons with the Adam optimizer and the ReLU activation function.
For the network of the discriminator, we adopt three-layer perceptrons with Adam optimizer.
The activation functions between layers are Leaky ReLU, while the activation function of output is the sigmoid activation function. The setting of main hyperparameters is shown in \tableautorefname~\ref{hyper}. 
\begin{table}[t]
    \centering
    \caption{The hyperparameters in the experiment}
    \begin{tabular}{ll}
    \toprule
    hyperparameters                         & value  \\ \midrule
    hidden size of actor network            & 256    \\
    hidden size of critic network           & 256    \\
    hidden size of discriminator network    & 128 \\ \bottomrule
    \end{tabular}
    \label{hyper}
\end{table}
\subsection{Tasks}
\paragraph{Squeeze} We present a discrete version of this problem. The state space is $\mathcal{S}=\{0, 1, 2\}$. Let $\mathcal{A}=\{0, 1\}$ denote the action space. The horizon of the environment is 3. The initial population state distribution is $\mu_0(s=2)=1$. The dynamic of the environment is given by:
\begin{align*}
    &P(s_{t+1}=1\mid s_t = \cdot, a=1) = \frac{3}{4},
    &P(s_{t+1}=0\mid s_t = \cdot, a=1) = \frac{1}{4}, \\
    &P(s_{t+1}=1\mid s_t = \cdot, a=0) = \frac{1}{4},
    &P(s_{t+1}=0\mid s_t=\cdot, a=0) = \frac{3}{4}
\end{align*}
The reward function is $$r(s, a, \mu)=\mathds{1}_{\{s=L\}}\mu(L)+\mathds{1}_{\{s=R\}}\mu(R).$$
\paragraph{RPS} The dynamic of RPS is deterministic:
\begin{align}
    P(s_{t+1}\mid s_t, a_t, \mu_t) = \mathds{1}_{s_{t+1}=a_t}
\end{align}
The state space $\mathcal{S}=\{C, R, P, S\}$ and the action space $\mathcal{A}=\{R, P, S\}$.
At the beginning of the game, all the agents are in the state $C$.
The reward function is shown in the following
\begin{align*}
    r(R, a, \mu_t) = 2\cdot\mu_t(S) - 1\cdot\mu_t(P)\\
    r(P, a, \mu_t) = 4\cdot\mu_t(R) - 2\cdot\mu_t(S)\\
    r(S, a, \mu_t) = 2\cdot\mu_t(P) - 1\cdot\mu_t(R)
\end{align*}
\paragraph{Flock} In nature, fish spontaneously align their velocity according to the overall movement of the fish school, resulting in a stable movement velocity for the entire school. We simplify this setting by defining a new dynamic as follows:
\begin{align*}
    x_{t+1}=x_{t}+v_{t} \Delta t
\end{align*} 
The action space $\mathcal{A}=\{0, 1, 2, 3\}$ corresponding to four directions of velocity with unit speed.
The reward is 
\begin{align*}
    f_{\beta}^{\mathrm{flock}}(x, v, u, \mu)=-\left\|\int_{\mathbb{R}^{2 d}} \left(v-v^{\prime}\right)\mathrm{d} \mu\left(x^{\prime}, v^{\prime}\right)\right\|^{2}
\end{align*}
\section{Comparison with MFCE derived by Muller et al.}\label{sec:comparison to muller}
In this section, We use the absent-minded driver game \cite{DBLP:conf/tark/PiccioneR96} to show the difference between AMFCE and the MFCE framework proposed by Muller et al. \cite{DBLP:journals/corr/abs-2208-10138}. Their notion of MFCE assumes that the mediator selects a mixed policy for the population and then sample a deterministic policy from the mixed policy and recommends to every agent, while our AMFCE framework assumes that the mediator selects a behavioral policy for the population at every time step and samples an action for every agent as recommendation. If agents are of bounded rationality, the mixed policy is not equivalent to the behavioral policy. 
\begin{example}\label{eg:example2}
    Suppose that the absent-minded driver game has two time steps. At the initial time, all the agents stay in state $s_1$. 
    The agent will stay in the state $s_1$ if action $B$ is chosen and the current population state distribution $\mu(s_1)=1$. If action $E$ is chosen, the agent will move to state $s_2$. If the agent enter the state $s_2$, the agent will stay in $s_2$ until the ending of the game. The reward function is
    \begin{align*}
        r(s, a, \mu)=\left\{\begin{array}{cc}
            3(1-\mu(s_1)), & a=E, s=s_1 \\
            \frac{1}{2}, & a=B, s=s_1, \mu=\cdot \\
            0, & otherwise
        \end{array}\right..
    \end{align*}

    Consider the case where the agents cannot remember the time step and the history. 
    The agent does not choose to take the deterministic policy of action $E$ at $s'$ because the policy makes the final payoff 0. 
    So the only MFCE policy in the game is the  deterministic policy to take action $B$ in any state, which has a final payoff of 1.

    On the other hand, we can find a possible AMFCE shown in the \tableautorefname~\ref{adaptive}. The agents will choose action $E$ if it is recommended.
\begin{table}[t]
\centering

\begin{tabular}{ccccc}
\toprule
Equilirbrium & \multicolumn{4}{c}{MFCE}                                              \\ \hline
Distribution & \multicolumn{2}{c}{$\pi(B|s',z=0)$} & \multicolumn{2}{c}{$\rho(z=0)$} \\ \hline
Value        & \multicolumn{2}{c}{1}               & \multicolumn{2}{c}{1}           \\ \hline
Equilirbrium & \multicolumn{4}{c}{AMFCE}                                             \\ \hline
Distribution & $\pi(B|s',z=0)$  & $\pi(B|s',z=1)$  & $\rho(z=0)$    & $\rho(z=1)$    \\ \hline
Value        & 1/2              & 1                & 1/2            & 1/2            \\ \bottomrule
\end{tabular}

\caption{The only MFCE and a possible AMFCE in the absent-minded driver game.}
\label{adaptive}
\end{table}
\end{example}

Example \ref{eg:example2} suggests that AMFCE has larger policy space than the MFCE proposed by Muller et al. \cite{DBLP:journals/corr/abs-2208-10138} because AMFCE assumes that the correlated signal sampled by the mediator corresponds to a behavioral policy.
\section{The Policy Set of MFCE Derived by Campi and Fisher.}\label{sec: policy_set}
Under the MFCE setting, each agent observes the full correlated signal sequence $\mathbf{z}$ and the recommended policy $\pmb{\pi}$ for the population.

Given $\pi$ and $\mathbf{z}$, each agent can compute the population distribution $\pmb{\mu}$ over the population's states and actions at time.

With knowledge of $\pmb{\mu}$, the agent can determine the expected utility for any action $a_t$ they might take at time $t$, given their state $s_t$.

The agent's goal is to maximize their expected cumulative return, which depends on both their own actions and the population distribution.

Therefore, the policy taken by the agent satisfies that $V_t\left(s,\pmb{\pi},{\color{black}\pmb{\mu}}\right)\geq V_t\left(s,\pmb{\pi}',\pmb{\mu}\right).$ for any policy $\pmb{\pi}'$, any time index $t\in \mathcal{T}$, and any initial state $s\sim \mu_0$.

Therefore, the policy set of MFCE is equivalent to MFNE policy set.
\section{Characterization of population distribution}\label{sec: def-sig}
\subsection{Introduction of signature}
\begin{definition}[Signature]
    Let $\mathbf{x}=\{x_1, \dots, x_L\}$ with $x_i\in\mathbb{R}^d$, for all $i$ and $L\ge2$. Denote $f:[0, 1]\to\mathbb{R}^{d}$ to be the continuous piecewise affine function such that $f(\frac{i-1}{L-1})=x_i$, $\forall i\in\{1, 2, \dots, L\}$.

    \begin{align}\label{eq:signature}
        \mathrm{Sig}(f)_{0, 1}=\left(1, M_1, \cdots, M_n, \ldots\right),
    \end{align}
where $M_n=\int_{s<s_{1}<\cdots<s_{n}<t}  \frac{\mathrm{d}f}{\mathrm{d}t}(s_1) \otimes \cdots \otimes \frac{\mathrm{d}f}{\mathrm{d}t}(s_n)\mathrm{d}t_1\cdots \mathrm{d}t_n$.

The signature of the path $\mathbf{x}$ is defined to be $\mathrm{Sig}(f)_{0,1}$, denoted as $\mathrm{Sig}(\mathbf{x})$.
\end{definition}
Signature of sequential data includes infinite terms as shown in the \cref{eq:signature}, but fortunately, terms $M_{n}$ enjoy factorial decay. In practice we select the first $n$ terms of the signature without losing crucial information of the data \cite{DBLP:conf/nips/KidgerBASL19}.
\subsection{Characterize population distribution using the signature}

In our framework, the mean field flow adheres to the McKean-Vlasov equation,
$$ \mu_{t+1}(s)=\sum_{a\in\mathcal{A}}\sum_{s'\in\mathcal{S}}\mu_t(s')P(s|s', a, \mu_t)\pi_t^E(a|s', z_t).$$
Given the environment and expert demonstrations, the transition probability $P(s|s', a, \mu_t)$, initial population state distribution $\mu_t$ and expert policy $\pi_t^E(a|s', z_t)$ are fixed. Therefore, $\mu_{t}(s)$ is determined by $z_{0:t}$ or $\hat{\mu_t}$ . We denote it as $\mu_t=\phi(z_{0:t})=\phi'({\rm Sig}(z_{0:t}))$.

Thus, we utilize ${\rm Sig}(z_{0:t})$ as the feature of $\mu_t$ input into the discriminator, thereby bypassing the need to estimate $\mu_t$ from the samples and avoiding additional estimation errors \cite{ramponi2023on}.

The signature offers an efficient encoding method for the sequence $z_{0:t}$ into a vector with fixed length. It's important to note that this choice is primarily for convenience. We could substitute the signature with alternative techniques for handling sequential data, such as Recurrent Neural Networks.


\end{document}